\documentclass[acmsmall,nonacm]{acmart}

\setcopyright{none}
\renewcommand\footnotetextcopyrightpermission[1]{}
\acmVolume{}
\acmNumber{}

\PassOptionsToPackage{capitalise,noabbrev}{cleveref}

\usepackage{thmtools, thm-restate}

\usepackage{amsfonts}

\usepackage{subcaption}

\usepackage{color}
\usepackage{stmaryrd}
\usepackage{multirow}

\makeatletter
\let\@realLabel\label
\makeatother

\RequirePackage{xparse}

\RequirePackage{xspace}
\def\SPDOffline{SPDOffline\xspace}

\NewDocumentCommand\UDP{o}{UDP\IfValueT{#1}{\textsubscript{#1}}\xspace}

\newcommand\Angle[1]{\langle#1\rangle}
\def\conc{\mathbin{\cdot}}
\def\mod{\mathbin{\%}}

\DeclareMathOperator\evts{evts}
\DeclareMathOperator\thd{thd}

\DeclareMathOperator\thds{thds}
\DeclareMathOperator\proj{proj}

\def\ba{\begin{array}}
\def\ea{\end{array}}
\def\bda{\begin{displaymath}\ba}
\def\eda{\ea\end{displaymath}}

\newcommand{\thread}[1]{\ensuremath{t_{#1}}} %%{\ensuremath{\tau_{#1}}} %%{t_{#1}}
\newcommand{\reqLockE}[1]{\reqE{#1}}
\newcommand{\lockE}[1]{\acqE{#1}} %% {\mathit{acq}(#1)} %%{\mathit{acq}(#1)}
\newcommand{\lockEb}[1]{\acqEb{#1}}
\newcommand{\unlockE}[1]{\relE{#1}} %% {\mathit{rel}(#1)} %%{\mathit{rel}(#1)}
\newcommand{\unlockEb}[1]{\relEb{#1}}
\newcommand{\eventE}[1]{e_{#1}}

\newcommand\evtAcc{a}

\newcommand\evtRel{r}

\definecolor{darkgray}{gray}{0.6}

\definecolor{GrayBgColor}{rgb}{0.9, 0.9, 0.9}
\newcommand{\BLOCKED}[1]{\underline{#1}} %%{\colorbox{GrayBgColor}{\ensuremath{#1}}}
\newcommand{\HIGHLIGHT}[1]{\colorbox{GrayBgColor}{\ensuremath{#1}}}

\newcommand{\LD}[3]{{\langle #1, #2, #3 \rangle}}
\newcommand\LDInduced[1]{\ensuremath{D_{#1}}}
\newcommand{\LSSym}{{\mathit LH}} %%{{\mathit LksHeld}}
\newcommand{\SET}[1]{\{ #1 \}}
\newcommand{\DP}[2]{{\mathit dp}_{#1} #2} %%{\circlearrowright^{#1}_{#2} #3}

\newcommand{\AllDP}[1]{\mathit{DP}_{#1}}

\newcommand{\standard}{std}

\newcommand{\multithreadindex}{MTi}
\newcommand{\all}{all}

\newcommand{\LHeld}[3]{\LSSym^{#2}_{#1}(#3)} %{LH_{#1}(#3)_{#2}}
\newcommand{\LHeldT}[1]{\LSSym_{#1}}

\newcommand{\StdLHeldT}{\LHeldT{\standard}}
\newcommand{\GeneralLHeld}[2]{\LSSym_{#1}(#2)}   %%{\LHeld{}{#1}{#2}}
\newcommand{\GeneralLHeldd}[1]{\LSSym_{#1}}   %%{\LHeld{}{#1}{#2}}

\newcommand{\LH}[2]{\LSSym_{#1}(#2)}
\newcommand{\LHH}[1]{\LSSym_{#1}}

\newcommand{\TOLHeldT}{\LHeldT{to}}

\newcommand{\LWLHeldT}{\LHeldT{lw}}

\newcommand{\ROLHeldT}{\LHeldT{ro}}

\newcommand{\HBLHeld}[2]{\LHeld{hb}{#1}{#2}}

\newcommand{\MultiLHeldIdxT}{\LHeldT{\multithreadindex}}

\newcommand{\GlobalLHeldT}{\LHeldT{\all}}

\newcommand{\PD}[2]{\DP{#1}{#2}\checkmark} %%{\mathit PD}(#3)^{#1}_{#2}}
\newcommand{\SPD}[3]{\DP{#1}{#2}{#3}\checkmark^{sync}} %%{\mathit PD}(#3)^{#1}_{#2}}
\newcommand{\AllPD}[1]{{\mathit PD}_{#1}}
\newcommand{\AllSPD}[2]{{\mathit SPD}^{#1}_{#2}}

\newcommand{\CS}[2]{{\mathit CS}_{#1}(#2)}
\newcommand{\CSS}[1]{{\mathit CS}}

\newcommand{\CSect}[4]{{\mathit CS}^{#2}_{#1}(#4)^{#3}} %{CS_{#1}(#4)^{#3}_{#2}}

\newcommand{\AcqRelPair}[2]{\Angle{#1,#2}}

\newcommand{\GeneralCSect}[2]{{\mathit CS}_{#1}(#2)}

\newcommand{\TOCSect}[3]{\CSect{to}{#1}{#2}{#3}}

\newcommand{\LKA}{\LK1}
\newcommand{\LKB}{\LK2}
\newcommand{\LKC}{\LK3}

\newcommand{\VA}{x}
\newcommand{\VB}{y}
\newcommand{\VC}{z}

\newcommand{\LK}[1]{l_{#1}}%{x_{{\alpha}_{#1}}}

\newcommand{\PO}{\mbox{P}}  %% partial order

\newcommand{\TO}{\mbox{TO}} %% thread-order
\newcommand{\LW}{\mbox{LW}} %% last-write
\newcommand{\RO}{\mbox{RO}} %% release-order
\newcommand{\HB}{\mbox{HB}} %% Lamport happens-before

\newcommand{\crpSymbol}{\mbox{\tiny crp}}
\newcommand{\crpLTSym}{<_{\crpSymbol}}
\newcommand{\crpLt}[1][]{\crpLTSym^{#1}}

\newcommand{\pSymbol}{\mbox{\tiny p}}
\newcommand{\pLTSym}{<_{\pSymbol}}
\newcommand{\pLt}[1][]{\pLTSym^{#1}}

\newcommand{\toSymbol}{\mbox{\tiny to}}
\newcommand{\toLTSym}{<_{\toSymbol}}
\newcommand{\toLt}[1][]{\toLTSym^{#1}}

\newcommand{\fjSymbol}{\mbox{\tiny fj}}
\newcommand{\fjLTSym}{<_{\fjSymbol}}

\newcommand{\lwSymbol}{\mbox{\tiny lw}}
\newcommand{\lwLTSym}{<_{\lwSymbol}}
\newcommand{\lwLt}[1][]{\lwLTSym^{#1}}

\newcommand{\roSymbol}{\mbox{\tiny ro}}
\newcommand{\roLTSym}{<_{\roSymbol}}
\newcommand{\roLt}[1][]{\roLTSym^{#1}}

\newcommand{\hbSymbol}{\mbox{\tiny hb}}
\newcommand{\hbLTSym}{<_{\hbSymbol}}
\newcommand{\hbLt}[1][]{\hbLTSym^{#1}}

\newcommand\indexedcap{\mathbin{\cap'}} %%{\mathbin{\cap^{idx}}}

\newcommand{\OK}{\mathit{ok}}
\newcommand{\FAIL}{\mathit{fail}}

\newcommand{\nf}[1]{\mbox{\normalfont #1}}

\newcommand{\uline}[1]{\rule[0pt]{#1}{0.4pt}}% Fill this blank
\newcommand{\dontCare}{\uline{.15cm}}

\newcommand{\incC}[2]{{\mathit inc}(#1,#2)}

\newcommand{\maxN}[2]{{\mathit max}(#1,#2)}

\newcommand{\accVC}[2]{#1[ #2 ]}

\newcommand{\VCConc}{\mathrel{||}}

\newcommand{\acqVC}[1]{{\mathit Acq(#1)}}

\newcommand{\threadVC}[1]{\textit{Th}(#1)}
\newcommand{\lastWriteVC}[1]{\ensuremath{L_W}(#1)}

\newcommand{\LDMapSym}{\ensuremath{D}}
\newcommand{\LDMap}[3]{\LDMapSym{\langle #1, #2, #3 \rangle}}

\newcommand{\PList}{\ensuremath{P}}

\RequirePackage{mytraces}

\makeatletter
\@ifdefinable\MakeTokenFromCsname{%
  \long\def\MakeTokenFromCsname#1#{%
    \romannumeral0\@MakeTokenFromCsname{#1}%
  }%
}%
\newcommand\@MakeTokenFromCsname[2]{%
  \expandafter\@@@MakeTokenFromCsname\expandafter{\csname#2\endcsname}{ #1}%
}%
\newcommand\@@@MakeTokenFromCsname[2]{#2#1}%
\makeatother

\makeatletter
\NewDocumentCommand\@DPO{m m o}{\mathrel{#2_{\text{\tiny #1}}\IfValueT{#3}{^{#3}}}}
\NewDocumentCommand\DPO{o m O{<} O{\leq} O{>} O{\geq} O{\parallel}}{%
  \MakeTokenFromCsname\NewDocumentCommand{#2Lt}{o}{\@DPO{\IfValueTF{#1}{#1}{#2}}{#3}[##1]}%
  \MakeTokenFromCsname\NewDocumentCommand{#2Leq}{o}{\@DPO{\IfValueTF{#1}{#1}{#2}}{#4}[##1]}%
  \MakeTokenFromCsname\NewDocumentCommand{#2Gt}{o}{\@DPO{\IfValueTF{#1}{#1}{#2}}{#5}[##1]}%
  \MakeTokenFromCsname\NewDocumentCommand{#2Geq}{o}{\@DPO{\IfValueTF{#1}{#1}{#2}}{#6}[##1]}%
  \MakeTokenFromCsname\NewDocumentCommand{#2Conc}{o}{\@DPO{\IfValueTF{#1}{#1}{#2}}{#7}[##1]}%
}
\makeatother

\DPO{Tr}
\DPO{PO}
\DPO{P}[\prec]
\DPO{LW}[\prec]
\DPO{PWR}[\prec]
\DPO[W]{WMHB}[\prec]
\DPO[C]{CMHB}[\prec]

\makeatletter
\def\DMO#1{%
  \MakeTokenFromCsname\DeclareMathOperator{@#1}{#1}%
  \MakeTokenFromCsname\NewDocumentCommand{#1}{o}{\MakeTokenFromCsname{@#1}\IfValueT{##1}{_{##1}}}%
}
\makeatother

\DMO{posOf}
\DMO{pos}
\DMO{spClosure}
\DMO{lwOf}
\DMO{lw}
\DMO{thdOf}
\DMO{evtsOf}
\DMO{thdsOf}
\DMO{acqOf}
\DMO{relOf}
\DMO{reqOf}
\DMO{acq}
\DMO{rel}
\DMO{req}
\DMO{lockOf}
\DMO{lock}
\DMO{crpOf}
\DMO{crp}
\DMO{spCrpOf}
\DMO{spCrp}

\renewcommand{\thdOf}{\thd}

\NewDocumentCommand\ld{mmmo}{(#1,#2,#3,\IfValueTF{#4}{\{#4\}}{\emptyset})}

\RequirePackage{listings}
\makeatletter % Redefinition of Description List Items source: http://tex.stackexchange.com/a/1248/13552
\newcommand\storedDescriptionLabel[1]{%
  \let\orglabel\label
  \let\label\@gobble
  \phantomsection
  \edef\@currentlabel{#1}%
  \let\label\orglabel
  #1%
}
\makeatother

\makeatletter
\NewDocumentCommand\@cond{m}{\mbox{\textsf{[#1]}}}
\NewDocumentCommand\cond{m}{\item[\@cond{#1}\label{#1}]}
\NewDocumentEnvironment{conditions}{O{}}{\begin{description}[format={\normalfont\storedDescriptionLabel},#1]}{\end{description}}
\NewDocumentEnvironment{conditions*}{O{}}{\begin{description*}[format={\normalfont\storedDescriptionLabel},mode=unboxed,#1]}{\end{description*}}
\makeatother

\makeatletter
\newcounter{listLengthCounter}
\NewDocumentCommand\listLength{m}{%
  \setcounter{listLengthCounter}{0}%
  \@for\@el:=#1\do{\stepcounter{listLengthCounter}}%
}
\makeatother

\makeatletter
\newcounter{currentCond}
\NewDocumentCommand\condRef{m}{%
  \listLength{#1}%
  \ifnum\c@listLengthCounter=1
    Condition~\ref{#1}%
  \else
    Conditions~%
    \setcounter{currentCond}{1}%
    \@for\@el:=#1\do{%
      \ifnum\c@currentCond>1
        \ifnum\c@currentCond=\c@listLengthCounter
          \ and~%
        \else
          ,
        \fi
      \fi
      \ref{\@el}%
      \stepcounter{currentCond}%
    }%
  \fi
}
\makeatother

\RequirePackage{myproofs}

\makeatletter
\newcounter{trace}
\NewDocumentCommand\traceNew{m}{\refstepcounter{trace}\@realLabel{trace:#1}T_{\thetrace}}
\NewDocumentCommand\traceRef{m}{\text{\hyperref[trace:#1]{$T_{\text{\ref*{trace:#1}}}$}}}
\makeatother

\newlength{\arraysep}
\RequirePackage{myalgos}

\makeatletter
\AddToHook{env/algorithmic/begin}{\def\@currentcounter{ALG@line}}
\makeatother
\RequirePackage{hyperref}
\RequirePackage{cleveref}
\crefalias{ALG@line}{line}
\crefname{line}{line}{lines}

\RequirePackage{thmtools}

\title{Beyond Per-Thread Lock Sets: Multi-Thread %%Covering
       Critical Sections and Dynamic Deadlock Prediction
      }
\author{Martin Sulzmann}
\affiliation{
  \institution{Karlsruhe University of Applied Sciences}
  \streetaddress{Moltkestrasse 30}
  \city{Karlsruhe}
  \postcode{76133}
  \country{Germany}
}
\email{martin.sulzmann@gmail.com}

\begin{document}

\begin{abstract}

Lock sets are commonly used for dynamic analysis of deadlocks.
The standard per-thread lock set construction only considers locks acquired in the same thread,
but is unaware of locks acquired in another thread.
This leads to false positives and false negatives.
The underlying issue is that the commonly used notion of a critical section
on which the lock set construction relies ignores
events from other threads.
We give a trace-based characterization of critical sections that drops this restriction. Critical sections are no longer restricted to a single thread
and can cover multiple threads. Such forms of critical sections
exist, are natural, and correct the standard formulation.

We show how to soundly approximate the trace-based characterization
via partial order relations.
Thus, we obtain an improved lock set construction that can still be efficiently computed
and allows us to remove false positives reported by the DIRK deadlock predictor
and remove false negatives by extending
the SPDOffline deadlock predictor.
We integrate various lock set constructions with increased precision in an extension of SPDOffline.
Our extensions remain sound (no false positives) but are more complete (fewer false negatives) w.r.t. SPDOffline.
%% For traces derived from C programs we can show that
%% SPDOffline misses to report a deadlock whereas our extensions catch the deadlock.
For an extensive standard benchmark suite we can also show that the performance is not affected.
\end{abstract}

\maketitle

\section{Introduction}
\label{sec:intro}

Bugs like deadlocks in concurrent software are notoriously hard to find
due to the non-deterministic nature of concurrent programs.
Our focus here is on resource deadlocks
in programs that make use of shared variables and locks
but without any message-passing primitives.
We follow the common assumption that
each release operation occurs in the same thread
as the prior acquire operation.
A resource deadlock arises if
there is a set of threads
where each thread in the set tries to acquire a lock that is held by some other thread in the set.

A well-studied line of work analyzes alternative thread interleavings
to  predict potential resource deadlocks from observed executions~\cite{10.5555/645880.672085,conf/hvc/BensalemH05,conf/pldi/JoshiPSN09,6718069,Samak:2014:TDD:2692916.2555262-previous,conf/ase/ZhouSLCL17,conf/oopsla/KalhaugeP18,conf/fse/CaiYWQP21,conf/pldi/TuncMPV23}.
The common method of these works is the search for \emph{deadlock patterns}.
A deadlock pattern represents a cyclic chain
among lock acquisitions in the observed execution
where the locks held by the acquires in the chain are disjoint.
The search for such patterns relies on the \emph{lock set} method~\cite{Dinning:1991:DAA:127695:122767};
the set of locks that are held (acquired) but not yet released at a certain point in the execution.
To the best of our knowledge,
all prior works assume that locks only protect operations in the acquiring thread.
This can lead to false deadlock warnings and even missed deadlocks
because locks can protect operations across thread boundaries.
We motivate the issue via the following two pseudo code examples
in~\cref{fig:pseudoDPFalsePositive} and~\cref{fig:pseudoDPFalseNegative}.

%% MS:
%% Joshi is not recent and likely is bogus anyway.
%% Only TR of SPD mentions "unbounded", let's discuss all this later.
%%
%% Use "request" instead of "event".

%% However, recent work suggests that this definition falls short of reality~\cite{conf/pldi/JoshiPSN09,conf/pldi/TuncMPV23}: locks can protect events across thread boundaries.

\begin{figure}[t]
    \begin{minipage}[b]{.3\linewidth}
      \begin{algorithmic}[1]
        \Procedure{main}{}
          \State $\forkE{\textsc{a}}$
          \State $\acqE{\vLock[1]}$
            \State $\vThd = \forkE{\textsc{b}}$
            \State $\joinE{\vThd}$
          \State $\relE{\vLock[1]}$
          \EndProcedure
        \algstore{false:main}
      \end{algorithmic}
    \end{minipage}%
    \hfill
    \begin{minipage}[b]{.3\linewidth}
      \begin{algorithmic}[1]
        \algrestore{false:main}
        \Procedure{a}{}
        \State $\acqE{\vLock[1]}$
        \State $\acqE{\vLock[2]}$
          \State $\acqE{\vLock[3]}$
          \State $\relE{\vLock[3]}$
          \State $\relE{\vLock[2]}$
          \State $\relE{\vLock[1]}$
          \EndProcedure
        \algstore{false:a}
      \end{algorithmic}
    \end{minipage}%
    \hfill
    \begin{minipage}[b]{.3\linewidth}
      \begin{algorithmic}[1]
        \algrestore{false:a}
        \Procedure{b}{}
          \State $\acqE{\vLock[3]}$
          \State $\acqE{\vLock[2]}$
          \State $\relE{\vLock[2]}$
          \State $\relE{\vLock[3]}$
        \EndProcedure
      \end{algorithmic}
  \end{minipage}%
  \caption{%
    Pseudo code program with false deadlock warning.
  }
  \label{fig:pseudoDPFalsePositive}
\end{figure}

\begin{figure}[t]
  \begin{minipage}[b]{.3\textwidth}
      \begin{algorithmic}[1]
        \Procedure{main}{}
        \State $\forkE{\textsc{a}}$
        \State $\acqE{\vLock[2]}$
        \State $\vThd = \forkE{\textsc{b}}$
        \State $\joinE{\vThd}$
        \State $\relE{\vLock[2]}$
        \EndProcedure
        \algstore{unrec:main}
      \end{algorithmic}
  \end{minipage}
  \hfill
  \begin{minipage}[b]{.3\textwidth}
    \begin{algorithmic}[1]
        \algrestore{unrec:main}
        \Procedure{a}{}
          \State  $\acqE{\vLock[1]}$
          \State  $\acqE{\vLock[2]}$
          \State  $\relE{\vLock[2]}$
          \State  $\relE{\vLock[1]}$
          \EndProcedure
        \algstore{unrec:a}
      \end{algorithmic}
  \end{minipage}
  \hfill
  \begin{minipage}[b]{.3\textwidth}
    \begin{algorithmic}[1]
        \algrestore{unrec:a}
        \Procedure{b}{}
        \State $\acqE{\vLock[1]}$
          \State $\relE{\vLock[1]}$
        \EndProcedure
      \end{algorithmic}
  \end{minipage}
  \caption{%
    Pseudo code program with unrecognized deadlock.
  } \label{fig:pseudoDPFalseNegative}
\end{figure}

\begin{figure}

  \begin{lstlisting}[language=C]
#include <pthread.h>

pthread_mutex_t l1 = PTHREAD_MUTEX_INITIALIZER;
pthread_mutex_t l2 = PTHREAD_MUTEX_INITIALIZER;

void* thread_l1_l2(void*) {
  pthread_mutex_lock(&l1);         // acquire l1
  pthread_mutex_lock(&l2);         // wants l2 -> potential deadlock here
  pthread_mutex_unlock(&l2);
  pthread_mutex_unlock(&l1);
  return NULL;
}

void* thread_l1(void*) {
  pthread_mutex_lock(&l1);        // wants l1 -> potential deadlock here
  pthread_mutex_unlock(&l1);
  return NULL;
}

int main() {
  pthread_t tid1, tid2;
  pthread_create(&tid2, 0, thread_l1_l2, 0);
  pthread_mutex_lock(&l2);         // acquire l2
  pthread_t tid3; pthread_create(&tid3, 0, thread_l1, 0);
  pthread_join(tid3, 0);           // waits for thread that wants l1
  pthread_mutex_unlock(&l2);
}
\end{lstlisting}

\caption{%
  Minimal C program corresponding to the pseudo code in~\cref{fig:pseudoDPFalseNegative}.
  Thread 1 acquires l1 then waits for l2. Thread 2 acquires l2, spawns a helper
  that tries to acquire l1, and then waits for it.
  This creates the multi-thread dependency that possibly leads to a deadlock.
  } \label{fig:real}
\end{figure}

%% More verbose version.
%% #include <pthread.h>
%%
%% pthread_mutex_t l1 = PTHREAD_MUTEX_INITIALIZER;
%% pthread_mutex_t l2 = PTHREAD_MUTEX_INITIALIZER;
%%
%% void* thread_l1_l2(void*) {
%%   pthread_mutex_lock(&l1);         // acquire x
%%   pthread_mutex_lock(&l2);         // wants y -> deadlock here
%%   pthread_mutex_unlock(&l2);
%%   pthread_mutex_unlock(&l1);
%%   return NULL;
%% }
%%
%% void* thread_l1(void*) {
%%   pthread_mutex_lock(&l1);
%%   pthread_mutex_unlock(&l1);
%%   return NULL;
%% }
%%
%% // Beyond per-thread lock set required
%% void* multi_thread_l2_l1(void*) {
%%   pthread_mutex_lock(&l2);         // acquire y
%%   pthread_t tid3; pthread_create(&tid3, 0, thread_l1, 0);
%%   pthread_join(tid3, 0);           // waits for thread that wants l1
%%   pthread_mutex_unlock(&l2);
%%   return NULL;
%% }
%%
%% int main() {
%%   pthread_t tid1, tid2;
%%   pthread_create(&tid2, 0, multi_thread_l2_l1, 0);
%%   pthread_create(&tid1, 0, thread_l1_l2, 0);
%%   pthread_join(tid1, 0); pthread_join(tid2, 0);
%% }

The pseudo code in \cref{fig:pseudoDPFalsePositive} illustrates
the issue of a false deadlock warning (adapted from \citep[Fig.\ 7]{report/MathurPTV23}).
Here, $\forkE{\textsc{p}}$ runs procedure~\textsc{p} in a new thread and returns the thread's identifier~$\vThd$, potentially to be used in $\joinE{\vThd}$ to block until the procedure is done.
Operations \textbf{$\acqE{\vLock}$} and \textbf{$\relE{\vLock}$}
acquire and release lock~$\vLock$.
We assume an execution where first the code in procedure~\textsc{a} is fully executed followed
by the remaining code in procedure~\textsc{main} and~\textsc{b}.

For this execution, we encounter a \emph{cyclic lock dependency}: procedure~\textsc{a} holds lock~$\vLock[2]$ while requesting lock~$\vLock[3]$, and procedure~\textsc{b} holds lock~$\vLock[3]$ while requesting lock~$\vLock[2]$.
A cyclic dependency between locks held and locks requested is identified, indicating a potential deadlock.
However, notice that procedure~\textsc{main} starts procedure~\textsc{b} while holding lock~$\vLock[1]$ until procedure~\textsc{b} is completed, and procedure~\textsc{a} holds lock~$\vLock[1]$ while requesting lock~$\vLock[3]$.
Hence, the cycle is broken because the threads in which procedures~\textsc{a} and~\textsc{b}
are executed hold the common lock~$\vLock[1]$.
Such locks are referred to as \emph{guard} locks.
This prevents the simultaneous execution of procedures~\textsc{a} and~\textsc{b}
and therefore there is no deadlock here.

The issue is that the thread in which procedure~\textsc{b} is executed does not hold lock~$\vLock[1]$.
Hence, existing approaches fail to detect guard lock~$\vLock[1]$
and might wrongly issue a deadlock warning.
This is the case for the DIRK deadlock predictor~\cite{conf/oopsla/KalhaugeP18}
as reported in \cite{report/MathurPTV23}.

%% MS: no need to see "how" to remove such false positives at this point
%%
%% Although the majority of deadlock-prediction methods is able to detect that lock~$\vLock[1]$ protects procedure~\textsc{a}, they fail to detect that lock~$\vLock[1]$ also protects procedure~\textsc{b}, leading to the false identification of a deadlock (i.e., a false positive).%
%% \footnote{\citet{conf/pldi/TuncMPV23} also do not consider procedure~\textsc{b} protected by lock~$\vLock[1]$, but they detect the guard lock by other means.}

%% MS:
%%
%% The occurrence of false positives due to missed guard locks in separate threads is known~\cite{conf/pldi/JoshiPSN09,conf/pldi/TuncMPV23}.
%% However, the assumption that locks can only protect events in the acquiring thread also leads to missed deadlocks.

Besides missed guard locks, we may fail to detect a cyclic lock dependency altogether which then leads to a missed deadlock.
Consider the pseudo code in \cref{fig:pseudoDPFalseNegative}.
Depending on the schedule among threads we might run here into a deadlock.
Consider an execution where the thread in which
procedure~\textsc{a} runs acquires lock~$\vLock[1]$
followed by procedure~\textsc{main} acquiring lock~$\vLock[2]$.
The request to acquire lock~$\vLock[2]$ in procedure~\textsc{a}
is blocked because procedure~\textsc{main} holds this lock.
Procedure~\textsc{main} starts procedure~\textsc{b} in a new thread
but is then blocked at the join statement because
procedure~\textsc{b} never terminates.
In detail, the request to acquire lock~$\vLock[1]$
in procedure~\textsc{b}
is blocked because procedure~\textsc{a} holds this lock.
We conclude that all threads are blocked and therefore we encounter
a deadlock.

The challenge for a dynamic analysis is to predict this deadlock
based on an execution that runs through without problems.
Consider an alternative schedule where the code in procedure~\textsc{a} is fully executed followed by execution of the remaining code
in procedure~\textsc{main} and~\textsc{b}.
The SPDOffline deadlock predictor~\cite{conf/pldi/TuncMPV23} fails to predict the deadlock
based on this execution
because lock dependencies are constructed on a per-thread basis.
SPDOffline detects that there is a lock dependency among~$\vLock[1]$
and~$\vLock[2]$ (same thread) in procedure~\textsc{a} but fails
to detect there is a lock dependency among~$\vLock[2]$
and~$\vLock[1]$ in procedure~\textsc{main}.
The reason is that procedure~\textsc{main} holds lock~$\vLock[2]$
but the subsequent request to acquire lock~$\vLock[1]$
is carried out in a new thread.

To underline that the above issues are not just theoretical,
\cref{fig:real} provides a minimal C implementation
using POSIX pthreads~\cite{Butenhof1997}
of the pseudo code program in~\cref{fig:pseudoDPFalseNegative}.
The complete runnable C program, together with additional C examples, is included in the anonymous artifact accompanying this submission.

%%Yet, as far as we know, no prior works detect a cycle of lock dependencies (including~\cite{conf/pldi/TuncMPV23}) and the deadlock is entirely missed (i.e., a false negative).
%%Moreover, the traditional notion of ``predictable'' deadlock---defined in terms of cyclic lock dependencies---does not catch the deadlock either, leading to (unintentionally) misleading precision claims.

%% MS:
%%
%% comment about ``predictable'' not clear at this point.
%% The issue is that ``predictable'' is underspecified.

\paragraph{Summary.}

Lock sets are used to construct deadlock patterns.
The standard per-thread lock set construction only considers locks acquired in the same thread. This restriction can lead to
false positives and false negatives
when using deadlock patterns to predict deadlocks in observed executions.
We address these limitations by introducing
corrected constructions of lock sets and deadlock patterns.

%% We address this restriction by introducing a new notion of \emph{multi-thread aware} lock sets, in which locks can protect requests across thread boundaries.
%% We integrate this idea  into the state-of-the-art deadlock-prediction algorithm \SPDOffline~\cite{conf/pldi/TuncMPV23}
%% to test the practicality of our approach.
%%Accordingly, we take a leap forward in deadlock prediction by developing \SPDUnbounded: an extension of the state-of-the-art sound deadlock-prediction algorithm \SPDOffline~\cite{conf/pldi/TuncMPV23} to account for unbounded lock dependencies.

%%Instead of adapting standard definitions of predictable deadlock to unbounded lock dependencies, we introduce a new notion of \emph{schedulable} deadlock based on a concept of ``stuckness'': alternative schedulings of events in which a set of lock requests can never be fulfilled.%
%%\footnote{Similar notions of stuckness for deadlock have been considered before, but we are the first to do so for execution traces.}
%%This way, we are able to independently assess the precision of \SPDUnbounded, proving the algorithm sound and more complete than \SPDOffline with respect to schedulable deadlocks.

%%To validate the feasibility of \SPDUnbounded in practice, we integrated the algorithm into the \UNDEAD deadlock predictor~\cite{conf/ase/ZhouSLCL17}.
%%Experiments on a large established set of benchmark traces show that \SPDUnbounded has insignificant overhead compared to \SPDOffline.

\paragraph{Our contributions and outline.}

Specifically, our contributions are:
\begin{itemize}
\item We give a trace-based characterization of critical sections
      and lock sets where we drop the restriction
      that locks can only protect events in the same thread (\cref{sec:lockset}).
\item We employ the corrected lock set definition to cover a larger class
      of deadlock patterns and predictable deadlocks.
      Important properties such as stuckness of predictable deadlocks
      are retained (\cref{sec:deadlock}).

\item We describe a family of lock set functions defined in terms of partial orders
      that soundly approximate the trace-based characterization of a lock set.
      The standard per-thread lock set construction is a member of this family represented by the thread-order relation.
      Under thread-order, two events are ordered based on their textual position but only if they are in the same thread.
      More expressive partial orders include last write dependencies and other forms of must-happen dependencies (\cref{sec:partial-order}).

\item We show how to compute the various lock set functions and
          the resulting deadlock patterns
          efficiently (\cref{sec:general-lockset-algo}).

\item We integrate the various lock sets and resulting deadlock patterns
      into the state of the art deadlock-prediction algorithm \SPDOffline~\cite{conf/pldi/TuncMPV23}.
      Our extensions of \SPDOffline retain soundness and
      covers a larger class of predictable deadlocks
      (\cref{sec:spd-ext}).

    \item We validate our extensions of \SPDOffline in practice.
      The measure the performance we consider a large set of standard benchmark traces.
      To measure the precision we consider traces derived from C programs (\cref{sec:experiments}).

\end{itemize}
Our implementation including benchmark traces and
simple C programs showing that critical sections covering multiple threads
exist is available at \url{https://zenodo.org/records/17310954}.

\noindent
\Cref{sec:overview} gives an overview of our work.
\Cref{sec:prelim} introduces basic definitions and notations.
\Cref{sec:rw} summarizes related work.
\Cref{sec:concl} concludes the paper.
Proofs for all results stated are given in the appendix.

\section{Overview}
\label{sec:overview}

\begin{figure}
\bda{@{}lcl}

%% \ba{@{}l@{}}
%% ex0
%% \ba{|l|l|l|}
%% \hline
%% \traceNew{LocksetsForDL} & \thread{1} & \thread{2}\\ \hline
%% \eventE{1}  & \forkE{\thread{2}}&\\
%% \eventE{2}  & \lockE{\LKA}&\\
%% \eventE{3}  & \lockE{\LKB}&\\
%% \eventE{4}  & \unlockE{\LKB}&\\
%% \eventE{5}  & \unlockE{\LKA}&\\
%% \eventE{6}  & &\lockE{\LKB}\\
%% \eventE{7}  & &\lockE{\LKA}\\
%% \eventE{8}  & &\unlockE{\LKA}\\
%% \eventE{9}  & &\unlockE{\LKB}\\
%%
%% \hline \ea{}
%% \ea

\ba{@{}l@{}}

 \ba{|l|l|l||l|l|}
 \hline
 \traceNew{LocksetsForDL} & \thread{1} & \thread{2} & \StdLHeldT  & \mbox{Lock dependencies} \\ \hline
\eventE{1}  & \forkE{\thread{2}}& &  &  \\
\eventE{2}  & \lockE{\LKA}& &   & \\
\eventE{3}  & \lockE{\LKB}&  &  \{ \LKA \}  &  \LD{\thread{1}}{\LKB}{\{\LKA\}}\\
\eventE{4}  & \unlockE{\LKB}& &  &  \\
\eventE{5}  & \unlockE{\LKA}& &  &  \\
\eventE{6}  & &\lockE{\LKB} &   & \\
\eventE{7}  & &\lockE{\LKA} &  \{ \LKB \} & \LD{\thread{2}}{\LKA}{\{\LKB \}} \\
\eventE{8}  & &\unlockE{\LKA} &   & \\
\eventE{9}  & &\unlockE{\LKB} &  &  \\

\hline \ea{}
\ea

\ \ \  & \ \ \

\ba{@{}l@{}}

\ba{|l|l|l|}
\hline
\traceNew{LocksetsForDLReordered} & \thread{1} & \thread{2}\\ \hline
\eventE{1}  & \forkE{\thread{2}}&\\
\eventE{2}  & \lockE{\LKA}&\\
\eventE{6}  & &\lockE{\LKB}\\

\eventE{3}  & \BLOCKED{\lockE{\LKB}}&\\
\eventE{7}  & &\BLOCKED{\lockE{\LKA}}\\

 \hline \ea{}

\ea

 \eda
 \caption{Trace \traceRef{LocksetsForDL} resulting from some program run (left),
   annotated with lock sets and resulting lock dependencies.
   Reordered trace \traceRef{LocksetsForDLReordered} exhibiting deadlock (right)
   where  $\BLOCKED{op}$ indicates a blocked operation.
}
   \label{fig:ex0}
\end{figure}

This section reviews the standard per-thread lock set construction known from the literature
and deadlock patterns which rely on lock sets.
We show that per-thread lock set construction has some inherent limitations
and explain our approach how to correct these limitations.

\subsection{Lock Sets for Deadlock Prediction}

We review the use of lock sets for deadlock prediction with the example in \cref{fig:ex0}.
The diagram on the left represents a program run
represented by a trace of events \traceRef{LocksetsForDL}.
We often represent traces in tabular notation
where each row contains exactly one event, in the column of the executing thread.
The textual order (from top to bottom) reflects the observed temporal order of events.

Events correspond to shared-memory/lock/thread operations executed by some thread.
Operation $\forkE{t}$ creates a new thread $t$ and
operation $\joinE{t}$ synchronizes with the termination of thread $t$.
We write $\LKA, \LKB, \LKC$ to denote locks and  $\VA,\VB,\VC$ to denote shared variables.
Operations $\lockE{\LK{}}/\unlockE{\LK{}}$ acquire/release lock $\LK{}$.
Operations $\readE{\VA}/\writeE{\VA}$ are shared memory read and write operations on $\VA$.
The same operation may appear multiple times in a trace, thus
we use indices as in $e_i$ to uniquely identify events in the trace.
Trace \traceRef{LocksetsForDL} can be more compactly represented
as the list $[e_1,e_2, e_3, e_4, e_5, e_6, e_7, e_8, e_9]$.

In trace \traceRef{LocksetsForDL} (left)
all events in thread $\thread{1}$ take place before the
events in thread $\thread{2}$, without apparent problem.
Based on trace \traceRef{LocksetsForDL} we reason about alternative schedules that
could have taken place and manifest a bug.
For example, we can swap the order among acquires on lock $\LKB$.
This then leads to the \emph{reordered trace prefix} $\traceRef{LocksetsForDLReordered} = [e_1,e_2,e_6]$ (right).
The highlighted acquire events $e_3$ and $e_7$ are not part of trace \traceRef{LocksetsForDLReordered}. Including either event would break the lock semantics by acquiring a lock already acquired but not yet released by the other thread.
Hence, we refer to $e_3$ and $e_7$ as blocked events.

To predict such deadlocks, the standard approach is to
derive from the trace a \emph{deadlock pattern}.
A deadlock pattern is a set $\SET{q_1,...,q_n}$ of events
where event $q_i$ represents the request to acquire
lock~$l_i$ in some thread $\thread{i}$
and the following two conditions are satisfied:
\begin{description}
\item[(DP-Cycle)] The lock acquired by $q_i$ is held by $q_{(i \mod n) + 1}$.
\item[(DP-Guard)] The locks held by each pair $q_i$, $q_j$ are disjoint.
\end{description}
Condition \textbf{DP-Cycle} states that there
is a cyclic chain of lock acquisitions and
condition \textbf{DP-Guard} states
that there is no common guard lock that prevents
the lock acquisitions to take place in some execution.
Crucially, both conditions rely on the notion
of \emph{locks held} (aka \emph{lock set})

The standard lock set definition in the literature is as follows.
Lock~$l$ is in the lock set of event~$e$, written
$l \in \StdLHeldT(e)$ if
\begin{description}
 \item[(CS-Enclosed)] $e$ is enclosed by acquire event~$\lockE{l}$
and release event~$\unlockE{l}$, and
\item[(CS-Match)]
    there are no other acquires and releases on~$l$ in between, and
\item[(CS-Per-Thread)] $e$ is in the same thread
  as $\lockE{l}$ and $\unlockE{l}$.
\end{description}
Thanks to condition~\textbf{CS-Per-Thread},
we can easily check conditions~\textbf{CS-Enclosed}
and~\textbf{CS-Match}
by considering the textual order of acquire-release
pairs~$\lockE{l}$ and~$\unlockE{l}$
that are in the same thread as~$e$.
This leads to the standard \emph{per-thread} lock set construction.

For trace \traceRef{LocksetsForDL} we record
the (per-thread) lock sets in some extra column.
We omit entries if they represent the empty set.
For example, for event $e_3$ that acquires
lock~$\LKB$ we find
the lock set $\StdLHeldT(e_3) = \{\LKA \}$ because
$e_3$ is surrounded by events~$e_2$ and $e_5$.

Based on the lock set information, we can derive
for each event~$e$ that acquires lock~$\LK{}$ a \emph{lock dependency}~\cite{conf/hvc/BensalemH05,conf/pldi/JoshiPSN09}
of the form $\LD{t}{\LK{}}{\StdLHeldT(e)}$
where $t$ refers to the thread that acquires lock $\LK{}$ and
$\StdLHeldT(e)$ is the
set of locks held by this thread while acquiring lock $\LK{}$.
For trace \traceRef{LocksetsForDL},
we obtain two lock dependencies:
$\LD{\thread{1}}{\LKB}{\{\LKA\}}$ and $\LD{\thread{2}}{\LKA}{\{\LKB\}}$.
The acquired lock $\LKB$ of the first dependency is in the lock set $\{ \LKB \}$
of the second dependency and and the acquired lock $\LKA$ of the second dependency
is in the lock set $\{ \LKA \}$ of the first dependency.
Hence, condition \textbf{DP-Cycle} is fulfilled.
The underlying lock sets are disjoint as $\{\LKA \} \cap \{ \LKB \} = \emptyset$.
Hence, condition \textbf{DP-Guard} is fulfilled.
This yields the deadlock pattern $\SET{e_3,e_7}$
and we issue a deadlock warning.

\begin{figure}
  \begin{subfigure}[t]{0.45\textwidth}
\bda{l}

%% latexTrace   $ addLoc ex2
\ba{|l|l|l|l|l|}
\hline
\traceNew{FPExample} & \thread{1} & \thread{2} & \thread{3} & \StdLHeldT \\ \hline
\eventE{1}  & \forkE{\thread{3}}&& & \\
\eventE{2}  & &&\lockE{\LKA} & \\
\eventE{3}  & &&\lockE{\LKB} & \\
\eventE{4}  & &&\lockE{\LKC} & \{\LKA, \LKB \} \\
\eventE{5}  & &&\unlockE{\LKC} & \\
\eventE{6}  & &&\unlockE{\LKB} & \\
\eventE{7}  & &&\unlockE{\LKA} & \\
\eventE{8}  & \lockE{\LKA}&& & \\
\eventE{9}  & \forkE{\thread{2}}&& &  \\
\eventE{10}  & &\lockE{\LKC}&&  \\
\eventE{11}  & &\lockE{\LKB}&& \{\LKC,  \HIGHLIGHT{\LKA} \} \\
\eventE{12}  & &\unlockE{\LKB}&&  \\
\eventE{13}  & &\unlockE{\LKC}& &  \\
\eventE{14}  & \joinE{\thread{2}}&&& \\
\eventE{15}  & \unlockE{\LKA}&&& \\

 \hline \ea{}
 \\
 \mbox{Lock dependencies}
 \\
 \LD{\thread{2}}{\LKB}{\{\LKC,\HIGHLIGHT{\LKA}\}}
 \ \ \ \LD{\thread{3}}{\LKC}{\{\LKA,\LKB \}}
 \eda
 % fontlock $
 \caption{False positive for deadlock prediction with per-thread lock
   set. Lock $\LKA$ not recognized.}
 \label{fig:ex2}
\end{subfigure}
\qquad
  \begin{subfigure}[t]{0.45\textwidth}
 \bda{l}

%% latexTrace   $ addLoc ex2_b
\ba{|l|l|l|l|l|}
\hline
\traceNew{FNExample} & \thread{1} & \thread{2} & \thread{3} & \StdLHeldT \\ \hline
\eventE{1}  & \forkE{\thread{3}}&& &\\
\eventE{2}  & \lockE{\LKB}&& &\\
\eventE{3}  & \forkE{\thread{2}}&& &  \\
\eventE{4}  & &\lockE{\LKA}& & \{ \HIGHLIGHT{\LKB} \}\\
\eventE{5}  & &\unlockE{\LKA}& &  \\
\eventE{6}  & \joinE{\thread{2}}&&&\\
\eventE{7}  & \unlockE{\LKB}&&&\\
\eventE{8}  & &&\lockE{\LKA}&\\
\eventE{9}  & &&\lockE{\LKB}& \{ \LKA \} \\
\eventE{10}  & &&\unlockE{\LKB}&\\
\eventE{11}  & &&\unlockE{\LKA}&\\

 \hline \ea{}
 \\
 \mbox{Lock dependencies}
 \\

\HIGHLIGHT{\LD{\thread{2}}{\LKA}{\{\LKB\}}}
 \ \ \ \LD{\thread{3}}{\LKB}{\{\LKA\}}

 \eda
% fontlock $
 \caption{False negative for deadlock prediction with per-thread lock set. Lock $\LKB$ not recognized.}
 \label{fig:ex2_b}
  \end{subfigure}
  \caption{Observed executions of programs from~\cref{fig:pseudoDPFalsePositive}
          and~\cref{fig:pseudoDPFalseNegative} represented as traces. Grey items
    are added by our corrected lock set definition.}
\label{fig:ex2-container}
\end{figure}

\subsection{Per-Thread Lock Set yields False Positives and False Negatives}

It is well-known that lock set-based methods are prone to false positives.
Numerous approaches attempt to reduce the number of false positives. For instance,
\citet{conf/oopsla/KalhaugeP18} were the first to
employ an SMT solver to exhaustively explore reorderings to eliminate infeasible false positives.
However, condition \textbf{CS-Per-Thread}
of the standard lock set construction is
an inherent source of false positives as well as false negatives
as motivated by the examples in~\cref{fig:pseudoDPFalsePositive}
and~\cref{fig:pseudoDPFalseNegative}.
We investigate the issue in more detail
and consider two specific executions represented
as traces in~\cref{fig:ex2-container}.
As before we record the lock sets in some extra column.
We only show the lock sets for some events.
For space reasons, resulting lock dependencies are
put underneath traces.

Trace \traceRef{FPExample} results from the
program in \cref{fig:pseudoDPFalsePositive}.
We use labels $\thread{1}$, $\thread{2}$ and $\thread{3}$
to refer to threads \textsc{Main}, \textsc{A} and \textsc{B}.
The standard method finds
dependencies $\LD{\thread{2}}{\LKB}{\{\LKC \}}$ and
$\LD{\thread{3}}{\LKC}{\{\LKA, \LKB \}}$ (ignore the grey $\LKA$ in the figure for now, it does not belong to the standard lock set).
Based on these dependencies we can argue that
$\SET{e_4,e_{11}}$ satisfies \textbf{DP-Cyle} and \textbf{DP-Guard}.
Hence, $\SET{e_4,e_{11}}$ represents a deadlock pattern and we issue a deadlock warning.

This warning is a false positive because events $e_{10} - e_{13}$ in thread $\thread{2}$ and
events $e_{3} - e_6$ in thread $\thread{3}$ are guarded by the common lock $\LKA$.
However, lock $\LKA$ is not part of the dependency $\LD{\thread{2}}{\LKC}{\{\LKA \}}$
because the standard lock set construction does not consider locks
acquired in a different thread.
See condition \textbf{CS-Per-Thread}.

It turns out that in every valid reordering, the acquire event $e_8$ will take place
before events $e_{10} - e_{13}$ and the matching release event $e_{15}$ will take place afterwards.
This example illustrates the need for a corrected definition of lock sets
where locks can protect requests across thread boundaries.
Then, the lock dependency constructed in $\thread{2}$
includes lock~$\{\LKA\}$. Thus,
the false positive is eliminated
because the lock sets are no longer disjoint.
Hence, \textbf{DP-Guard} is violated.

The construction of deadlock patterns based on standard lock sets
also leads to false negatives.
Consider trace \traceRef{FNExample} resulting from
the program in \cref{fig:pseudoDPFalseNegative}.
The program contains a  deadlock, but
standard lock sets only yield the dependency
$\LD{\thread{3}}{\LKB}{\{\LKA\}}$.
Hence, there is no deadlock pattern and no deadlock warning is issued.
The corrected lock set construction
adds lock $\LKB$ to the lock set of event~$e_4$.
This then results in the dependency
$\LD{\thread{2}}{\LKA}{\{\LKB\}}$ which gives rise to a cyclic dependency
with $\LD{\thread{3}}{\LKB}{\{\LKA\}}$ where
the underlying lock sets are disjoint.
Thus, $\SET{e_4,e_9}$ satisfies \textbf{DP-Cyle} and \textbf{DP-Guard}
and we correctly issue that there is a deadlock.

\subsection{Beyond Per-Thread Lock Sets}

We correct the standard lock set definition
by dropping condition \textbf{CS-Per-Thread}.
Without condition \textbf{CS-Per-Thread},
the check \emph{$e$ is enclosed by ~$\lockE{l}$ and~$\unlockE{l}$}
in condition~\textbf{CS-Enclosed} can no longer be carried
out by considering the textual order among events.
To ensure that event~$e$ is protected by lock~$l$,
we demand that under \emph{all}
alternative scheduling of events, and for each schedule,
event~$e$ is enclosed by acquire event~$\lockE{l}$
and release event~$\unlockE{l}$ with no other acquires and releases
on~$l$ in between.
This then allows to discover locks that protect events across thread boundaries.
Thus, the false positive and false negative in the above examples can be removed.
%% MS: omit, sounds too negative
%% However,  the construction of lock sets becomes more involved.

As there can be exponentially many alternative schedulings,
the above corrected lock set definition will not lead to a practical algorithm.
To obtain an efficient algorithm, we approximate the set of alternative
schedulings via partial orders derived from the trace.
This is something well-explored in the data race prediction setting~\cite{lamport1978time,Smaragdakis:2012:SPR:2103621.2103702,Roemer:2018:HUS:3296979.3192385,conf/pldi/KiniMV17,10.1145/3360605,conf/mplr/SulzmannS20}.
Given some partial order~$<$ among events in the trace,
lock~$l$ is in $e$'s lock set if
$\lockE{l} < e < \unlockE{l}$.

In the data race setting, partial orders usually
\emph{overapproximate}.
If two events $e$ and $f$ are ordered
this does not necessarily mean that $f$ cannot appear before $e$
in some alternative schedule.
But if $e$ and $f$ are unordered then we can guarantee
that there must be an alternative schedule
where $e$ and $f$ appear right next to each other.
This then satisfies the condition for a data race
and ensures that only true positive data races are predicted.

In our setting, we require partial orders that \emph{underapproximate}.
If $e < f$ then we can guarantee that $e$ must appear before $f$
in any alternative schedule of events.
But if $e$ and $f$ are unordered it could be the case that
$e$ always appears before $f$.
Underapproximation guarantees that
if $\lockE{l} < e < \unlockE{l}$,
then~$e$ is is enclosed by~$\lockE{l}$ and~$\unlockE{l}$
under all alternative schedules.

The most simple underapproximating partial order is
the \emph{thread-order} relation.
We define $e \toLt[T] f$ if $e$ appears before $f$ in
trace~$T$ and $e$ and $f$ are in the same thread.
Effectively, this then yields the standard per-thread
lock set construction.
But there is ample space for further underapproximating partial orders.
For example, the \emph{fork-join-order} relation~$\fjLTSym$
that extends~$\toLTSym$ by including fork-join dependencies.

We will consider a number of different partial orders
suitable for our purposes (\cref{sec:partial-order})
and show how to effectively compute lock sets and deadlock patterns
based on these partial orders (\cref{sec:general-lockset-algo}).
We integrate the various lock set constructions into the state of the art
deadlock predictor \SPDOffline~\cite{conf/pldi/TuncMPV23}
(\cref{sec:spd-ext}).
The resulting extensions of \SPDOffline\ remain sound (no false positives),
become more complete (fewer false negatives) and performance is not affected
(\cref{sec:experiments}).

\section{Preliminaries}
\label{sec:prelim}

We introduce some standard notations, following
related work in predictive analyses of concurrent programs~\cite{sen2005detecting,Smaragdakis:2012:SPR:2103621.2103702,conf/pldi/TuncMPV23}.

\subsection{Events and Traces}
\label{sec:prelim:trace}

We consider dynamic analyses which observe the behavior
of a concurrent program represented as a trace.
A trace is a list of events reflecting a single execution of a concurrent program under the sequential consistency memory model~\cite{Adve:1996:SMC:619013.620590}.
Each event is connected to some concurrency primitive
such as acquire/release of locks (mutexes) and writing/reading of shared variables.
We use $\vLock[1],\vLock[2],\ldots$ and $x,y,z$ for lock and shared variables, respectively.

\begin{definition}[Events and Traces]
  \label{def:run-time-traces-events}
  \bda{@{}l@{}c@{}lll@{}c@{}llr@{}}
    \alpha,\beta,\delta & {} ::= {} & 1 \mid 2 \mid \ldots & \text{(unique event ids)}
    & \vEvt & {} ::= {} & (\alpha,\vThd,op) & \text{(events)}
    \\ t,s,u & {} ::= {} & \tid{1} \mid \tid{2} \mid \ldots & \text{(thread ids)}
    & \vTr & {} ::= {} & [] \mid e : T & \text{(traces)}
    \\ op & {} ::= {} & \readE{x} \mid \writeE{x} \mid \reqE{\vLock} \mid \acqE{\vLock} \mid \relE{\vLock}
    \span\span\span\span
    %  \mid \forkE{t}
    %  \mid \joinE{t}
    & \text{(operations)}
  \eda
\end{definition}

\noindent
We write $[o_1,\dots,o_n]$
for a list of objects as a shorthand of $o_1:\dots:o_n:[]$ and use the
operator  ``${}\conc{}$'' for list concatenation.

An event~$\vEvt$ is represented by a triple~$(\alpha,\vThd,op)$
where
$\alpha$ is a unique event identifier,
$op$ is an operation,
and
$\vThd$ is the id of the thread in which the operation took place.
The main thread has thread id~$\tid{1}$.
The unique event identifier unambiguously identifies events under trace reordering.
We often omit identifier and/or thread when denoting events, and omit parentheses if only the operation remains, writing $\vEvt = (\vThd,op)$ or $\vEvt = op$ instead of $\vEvt = (\alpha,\vThd,op)$.

The operations~$\readE{x}$ and~$\writeE{x}$ denote reading of and writing to a shared variable~$x$, respectively.
Operations~$\acqE{\vLock}$ and~$\relE{\vLock}$ denote acquiring and releasing a lock~$\vLock$, respectively.
We also include lock requests~$\reqE{\vLock}$ to denote the (possibly unfulfillable) attempt at acquiring a lock~$\vLock$, following, e.g.,~\cite{conf/oopsla/KalhaugeP18}.
For brevity, we omit operations $\forkE{t}$ and $\joinE{t}$.
% We write~$\forkE{t}$ for the creation of a new thread with thread id~$t$.
% We write~$\joinE{t}$ for a join with a thread with thread id~$t$.

Our tabular notation for traces has one column per thread. The events
for a thread are lined up in the thread's column and the trace
position corresponds to the row number.
See trace \traceRef{FNExample} in \cref{fig:ex_mt3b}.
We model fork and join events via shared variable operations.
For example, events $[e_2,e_3]$ emulate $\forkE{\thread{2}}$
and events $[e_7,e_8]$ emulate $\joinE{\thread{2}}$.
For brevity, we do not protect the read/write operations by some locks.
Acquires $e_5$ and $e_{12}$ are preceded by some lock request events.
For brevity, we omit lock request events for acquires $e_1$ and $e_{10}$.

Notation~$\thd(\vEvt)$ extracts the thread id from an event.
We define~$\lock(e) = \vLock$
if $e = (t,\lockE{\vLock}$, $e=(t,\unlockE{\vLock})$ or $e=(t,\reqE{\vLock})$.
We write $\vEvt \in \vTr$ to indicate that $\vTr = [\vEvt[1],\ldots,\vEvt[n]]$ and
for some $k\in\{1,...,n\}$ we find that $\vEvt = \vEvt[k]$, defining also $\pos[\vTr](\vEvt) = k$.
The set of events in a trace is then $\evts(\vTr) = \{ \vEvt \mid \vEvt \in \vTr \}$, and the set of thread ids in a trace is $\thds(\vTr) = \{ \thd(\vEvt) \mid \vEvt \in \vTr \}$.
For trace~$\vTr$ and events~$\vEvtA,\vEvtB \in \evts(\vTr)$, we define trace order: $\vEvtA \TrLt[\vTr] \vEvtB$ if~$\pos[\vTr](\vEvtA) < \pos[\vTr](\vEvtB)$.

\subsection{Well Formedness}
\label{sec:prelim:wf}

Traces must be well formed, following the standard sequential consistency conditions for concurrent objects \cite{Adve:1996:SMC:619013.620590,Said:2011:GDR:1986308.1986334,Huang:2014:MSP:2666356.2594315}.

\begin{definition}[Well Formedness]
  \label{def:WF}
  Trace~$\vTr$ is \emph{well formed} if all the following conditions are satisfied:
  \begin{description}
    \item[(WF-Acq)]
      For every~$\vAcq = (\vThd,\acqE{\vLock}), \vAcq' = (\vThd',\acqE{\vLock}) \in \vTr$ where~$\vAcq \TrLt[\vTr] \vAcq'$, there exists~$\vRel = (\vThd,\relE{\vLock}) \in \vTr$ such that~$\vAcq \TrLt[\vTr] \vRel \TrLt[\vTr] \vAcq'$.

    \item[(WF-Rel)]
      For every~$\vRel = (\vThd,\relE{\vLock}) \in \vTr$, there exists~$\vAcq = (\vThd,\acqE{\vLock}) \in \vTr$ such that~$\vAcq \TrLt[\vTr] \vRel$ and there is no~$\vRel' = (\vThd',\relE{\vLock}) \in \vTr$ with $\vAcq \TrLt[\vTr] \vRel' \TrLt[\vTr] \vRel$.
      We define~$\rel[\vTr](\vAcq) = \vRel$ and~$\acq[\vTr](\vRel) = \vAcq$.

    \item[(WF-Req)]
      For every~$\vAcq = (\vThd,\acqE{\vLock}) \in \vTr$, there exists~$\vReq = (\vThd,\reqE{\vLock}) \in \vTr$ such that $\vReq \TrLt[\vTr] \vAcq$
      and there is no event $e \in \vTr$ such that $\vReq \TrLt[\vTr] e \TrLt[\vTr] \vAcq$ and $\thd(e) = \vThd$.
      For every~$\vReq = (\vThd,\reqE{\vLock}) \in \vTr$,
      if $e \in \vTr$ where $\vReq \TrLt[\vTr] e$  and $\thd(e) = \vThd$
      and there is no other $f \in \vTr$ such that
      $\vReq \TrLt[\vTr] f \TrLt[\vTr] e$  and $\thd(f) = \vThd$,
      then $e=(t,\lockE{\vLock})$.

    %  \item[Fork-1:]
    %    For each $t \ne t_1$ there exists at most one $(\dontCare,\forkE{t}) \in T$ and $(\dontCare,\forkE{t_1}) \notin T$.

    %  \item[Fork-2:]
    %    For each $\evtAA = (t,op) \in T$ where $t \ne t_1$ there exists $\evtBB = (s,\forkE{t}) \in T$ where $\evtBB \TrLt \evtAA$.

    %  \item[Join:]
    %    For each $\evtAA = (s,\joinE{t}) \in T$ we have $s \ne t$ and for all $\evtBB = (t,op) \in T$ we have $\evtBB \TrLt \evtAA$.
  \end{description}
\end{definition}

\noindent
\textbf{WF-Acq} states that a previously acquired lock can only be acquired after it has been released.
Similarly, \textbf{WF-Rel} states that a lock can only be released after is has been acquired but not yet released.
Note that these conditions require matching acquires and releases to occur in the same thread.

\textbf{WF-Req} states that each acquire must be immediately preceded by a lock request.
Requests do not have to be fulfilled by a corresponding acquire event.
However, if there is no acquire that fulfills a request, no further events
can come after the request in that thread.
Request events are not strictly necessary but simplify up-coming
definitions that reason about `blocked' operations.
For brevity, we often omit requests from example traces.
For example, see \traceRef{FNExampleLW} in \cref{fig:ex_mt3b}.
For such cases we assume that requests precede each acquire implicitly.

%% MS: nope, we're more restricted here, for convenience.
%% At the end of the trace, acquired locks do not have to be released and requests do not have to be fulfilled.

% Condition \textbf{Fork-1} states that a thread can be created at most once.
% Condition \textbf{Fork-2} states that each thread except the main thread is preceded by a fork event.
% Both conditions imply that for each event ${(\beta, t,\forkE{s})}$ we have $t \ne s$.

% Condition \textbf{Join} states that all events from a joined thread appear before the join event.
% There can be several join events $\joinE{t}$ for the same thread $t$.
% A join operation $\joinE{t}$ does not necessarily have to appear in the thread that forked thread $t$.

For brevity, we often drop the \emph{well-formed} assumption
imposed on traces. All traces we consider here
are well-formed.

\subsection{Trace Reorderings}
\label{sec:prelim:crp}

Trace-based dynamic analyses aim to expose bugs based
on some program run represented as a trace~$T$.
Bugs such as deadlocks rarely show up directly.
Hence, we wish to \emph{predict} a deadlock
by considering some alternative scheduling of the events in $T$ that
lead to a deadlock.
Naively, we could consider all permutations of the original trace~$T$.
However, not all such reorderings are feasible in the sense of being reproducible by executing the program with a different schedule.
In addition to well formedness, a reordering must guarantee that (a)~program order and (b)~last writes are maintained.
A reordering maintains program order if, in any thread, the order of events is unchanged.
For guarantee~(b), note that every read observes some write on the same variable: the last preceding such write.
Last writes are then maintained if the write observed by any read is unchanged.
Guarantee~(b) is particularly important to ensure that the control flow of the traced program is unaffected by the reordering, e.g., when the read is used in a conditional statement.
%% \bh{Stress that we assume that reading a different value leads to deadlock.}
%% MS: done via example ex_mt3b

Reorderings do not have to run to completion, so we consider reordered \emph{prefixes} of traces.

\begin{definition}[Correctly Reordered Prefix]
  \label{def:crp}
  The \emph{projection} of~$\vTr$ onto thread~$\vThd$, denoted $\proj(\vTr,\vThd)$, restricts~$\vTr$ to events~$\vEvt$ with $\thd(\vEvt) = \vThd$.
  That is,
  $\vEvt \in \proj(\vTr,\vThd)$ if and only if $\vEvt \in \vTr$ and $\thd(\vEvt) = \vThd$,
  and
  $\vEvtA \TrLt[\proj(\vTr,\vThd)] \vEvtB$ implies $\vEvtA \TrLt[\vTr] \vEvtB$.

  Take~$\vEvtA = \readE{x},\vEvtB = \writeE{x} \in \vTr$.
  We say that~$\vEvtB$ is the \emph{last write} of~$\vEvtA$ in~$\vTr$, denoted $\vEvtB = \lw[\vTr](\vEvtA)$, if $\vEvtB$ appears before~$\vEvtA$ with no other write on~$x$ in between.
  That is,
  $\vEvtB \TrLt[\vTr] \vEvtA$,
  and
  there is no~$\vEvtC = \writeE{x} \in \vTr$ such that $\vEvtB \TrLt[\vTr] \vEvtC \TrLt[\vTr] \vEvtA$.

  Trace~$\vTr'$ is a \emph{correctly reordered prefix} of~$\vTr$ if the following
  conditions are satisfied:
  \begin{description}
    \item[(CRP-WF)] $\vTr'$ is well formed and $\evts(\vTr') \subseteq \evts(\vTr)$.
    \item[(CRP-PO)] For every~$\vThd \in \thds(\vTr')$, $\proj(\vTr',\vThd)$ prefixes~$\proj(\vTr,\vThd)$.
    \item[(CRP-LW)] For every~$\vEvtA = \readE{x} \in \vTr'$ where $\vEvtB = \lw[\vTr](\vEvtA)$ we have that $\vEvtB = \lw[\vTr'](\vEvtA).$
  \end{description}
  We write $\crp(\vTr)$ to denote the set of correctly reordered prefixes of~$\vTr$.
\end{definition}

\noindent
Condition~\textbf{CRP-LW} follows the formulation in~\cite{Mathur:2018:HFR:3288538:3276515} and uses function~$\lw[]$ to guarantee that last writes are maintained under reordering.
There are other formulations~\cite{conf/pldi/TuncMPV23} that
make use of a "reads-form" function which is just a different name for
function~$\lw[]$.

\begin{figure}[t]
  \begin{subfigure}[t]{0.45\textwidth}
    \bda{l}
%% latexTrace $ addLoc ex_mt3b
\ba{|l|l|l|l|}
\hline
\traceNew{FNExampleLW} & \thread{1} & \thread{2} & \thread{3}\\ \hline
\eventE{1}  & \lockE{\LKA}&&\\
\eventE{2}  & \writeE{\VA}&&\\
\eventE{3}  & &\readE{\VA}&\\
\eventE{4}  & &\reqLockE{\LKB}&\\
\eventE{5}  & &\lockE{\LKB}&\\
\eventE{6}  & &\unlockE{\LKB}&\\
\eventE{7}  & &\writeE{\VB}&\\
\eventE{8}  & \readE{\VB}&&\\
\eventE{9}  & \unlockE{\LKA}&&\\
\eventE{10}  & &&\lockE{\LKB}\\
\eventE{11}  & &&\reqLockE{\LKA}\\
\eventE{12}  & &&\lockE{\LKA}\\
\eventE{13}  & &&\unlockE{\LKA}\\
\eventE{14}  & &&\unlockE{\LKB}\\

 \hline \ea{}
    \eda
 \caption{Reformulation of trace \traceRef{FNExample} in \cref{fig:ex2_b}.}
 \label{fig:ex_mt3b}
  \end{subfigure}
      \qquad
  \begin{subfigure}[t]{0.4\textwidth}
    \bda{l}
\ba{|l|l|l|l|}
\hline
\traceNew{FNExampleLWCRP} & \thread{1} & \thread{2} & \thread{3}\\ \hline
\eventE{1}  & \lockE{\LKA}&&\\
\eventE{10}  & &&\lockE{\LKB}\\
\eventE{2}  & \writeE{\VA}&&\\
\eventE{3}  & &\readE{\VA}&\\
\eventE{4}  & &\reqLockE{\LKB}&\\
\eventE{11}  & &&\reqLockE{\LKA}\\
\hline \ea{}
    \eda
     \caption{Reordered prefix of \traceRef{FNExampleLW}.}
 \label{fig:ex_mt1b}
  \end{subfigure}

  \caption{Request events, fork-join via last-write and correct reorderings.}
\label{fig:ex_prelim}
\end{figure}

\begin{example}
Consider traces $\traceRef{FNExampleLWCRP} \in \crp(\traceRef{FNExampleLW})$ in \cref{fig:ex_prelim}.
The explicit request events $e_4$ and $e_{11}$ in \traceRef{FNExampleLWCRP} express
the fact that thread~$\thread{2}$ attempts to acquire lock~$l_2$
and thread~$\thread{3}$ attempts to acquire lock~$l_1$.
However, both requests cannot be fulfilled
because \traceRef{FNExampleLWCRP} extended with either $e_5$ or $e_{12}$ violates
lock semantics.
Further note that \traceRef{FNExampleLWCRP} cannot be extended
with $e_8$ either as this violates the last-write condition~\textbf{CRP-LW}.
\end{example}

\section{Critical Section and Lock Set}
\label{sec:lockset}

We consider a specific program run represented by some trace~$T$.
We assume that trace~$T$ is a well-formed trace with a matching release
event for each acquire. This assumption imposes no
restriction as we can always add missing release events at the
end of a trace.
In trace~$T$ we wish to identify critical sections that
protect events.
Informally, an event is part of a critical section if it is enclosed
by an acquire event and its matching release.
Recall the informal conditions~\textbf{LS-Enclosed} and~\textbf{LS-Match} form \cref{sec:overview}.
We present a trace-based definition capturing this requirement
by considering all alternative schedules
represented by the set of correctly reordered prefixes.

\begin{figure}

  \hspace{-.8cm}
  \begin{subfigure}[t]{0.3\textwidth}
    \bda{l}
%%    latexTrace $ addLoc bynd_mt_1b
\ba{|l|l|l|l|}
\hline
\traceNew{TraceCSLW} & \thread{1}    & \thread{2} & \thread{3}\\ \hline
 \eventE{1}  &  \lockEb{\LKA}  &&\\
\eventE{2}  & \writeEb{\VA}&&\\
\eventE{3}  & \writeEb{\VB}&&\\
\eventE{4}  & & \readEb{\VA}&\\
\eventE{5}  & && \readEb{\VB}\\
\eventE{6}  & & \writeEb{\VA} &\\
\eventE{7}  & && \writeEb{\VB}\\
\eventE{8}  & \readEb{\VA}&&\\
\eventE{9}  & \readEb{\VB}&&\\
\eventE{10}  & \unlockEb{\LKA}&&\\

 \hline \ea{}
    \eda
    \caption{} %%All events in $\thread{2}$ and $\thread{3}$ protected.
      %%Last-write dependencies.
 \label{fig:ex-cs-lw}
  \end{subfigure}
  \qquad
  \begin{subfigure}[t]{0.23\textwidth}
    \bda{l}
%%latexTrace $ addLoc bynd_mt_2
\ba{|l|l|l|}
\hline
\traceNew{TraceCSRO} & \thread{1} & \thread{2}\\ \hline
\eventE{1}  & \lockE{\LKB}&\\
\eventE{2}  & \writeE{\VA}&\\
\eventE{3}  & \lockEb{\LKA}&\\
\eventE{4}  & \unlockEb{\LKB}&\\
\eventE{5}  & &\lockE{\LKB}\\
\eventE{6}  & & \readEb{\VA}\\
\eventE{7}  & & \unlockEb{\LKB}\\
\eventE{8}  & & \writeEb{\VB}\\
\eventE{9}  &   \readEb{\VB}&\\
\eventE{10}  &  \unlockEb{\LKA}&\\

 \hline \ea{}
    \eda
    \caption{} %%Some events in $\thread{2}$ protected.
      %% Release-order dependency.
 \label{fig:ex-cs-ro}
  \end{subfigure}
  \qquad
  \begin{subfigure}[t]{0.3\textwidth}
    \bda{l}
%% latexTrace $ addLoc bynd_mt7c
\ba{|l|l|l|l|}
\hline
\traceNew{TraceCSROBYND} & \thread{1} & \thread{2} & \thread{3}\\ \hline
\eventE{1}  & \lockE{\LKB}&&\\
\eventE{2}  & \writeE{\VC}&&\\
\eventE{3}  & \lockEb{\LKA}&&\\
\eventE{4}  & \unlockEb{\LKB}&&\\
\eventE{5}  & &\lockE{\LKB}&\\
\eventE{6}  & &\writeE{\VB}&\\
\eventE{7}  & &&\readE{\VB}\\
\eventE{8}  & &&\readEb{\VC}\\
\eventE{9}  & &&\writeEb{\VA}\\
\eventE{10}  & &\readEb{\VA}&\\
\eventE{11}  & &\unlockEb{\LKB}&\\
\eventE{12}  & &\writeEb{\VA}&\\
\eventE{13}  & \readEb{\VA}&&\\
\eventE{14}  & \unlockEb{\LKA}&&\\

 \hline \ea{}
    \eda
 \caption{} %%No events in $\thread{2}$ protected.}
 \label{fig:ex-cs-not}
\end{subfigure}

  \caption{Events that are part of a critical section involving lock~$\LKA$ are
         written in {\bf bold face}.}
\label{fig:ex-cs}
\end{figure}

\subsection{Critical Section}

We say event $f$ \emph{must be preceded}
by event~$e$ \emph{under all trace reorderings},
written $e \crpLt[T] f$, if $\forall T' \in \crp(T)$ where $f \in T'$
we have that $e \in T'$ and $e \TrLt[T'] f$.

\begin{definition}[Critical Section]
\label{def:cs}
  Suppose there are events
  $a = (t, \lockE{l}), r = (t, \unlockE{l}) \in T$ for
  some lock~$l$.
  We say that $e \in T$ is in the
  \emph{critical section enclosed by acquire $a$
    and release $r$ under all schedules},
  written
  $ e \in \GeneralCSect{T}{a,r}$
  if the following two conditions are met:

  \begin{description}
  \item[(CS-Enclosed)]  $a \crpLt[T] e$ and $e \crpLt[T] r$.

\item[(CS-Match)] There is no release event $r'= (t, \unlockE{l}) \in T$  and
       $T' \in \crp(T)$
     such that $a \TrLt[T'] r'$ and~$r' \TrLt[T'] e$.
  \end{description}

\end{definition}
Condition \textbf{CS-Enclosed}  states that event~$e$
must be preceded by acquire~$a$ and release~$r$ must be preceded by event~$e$.
This means that event~$e$ is enclosed
by the critical section represented by acquire~$a$ and release~$r$
and captures the intuitive understanding that event~$e$ is protected by lock~$l$.
Condition \textbf{CS-Match} guarantees that $r$ is the matching release for the acquire event $a$.

\subsection{Lock Set}

Based on the above characterization of critical section
we can define the notion of a lock set.
Lock $l$ is held by some event $e$ if the event is part of a critical section for lock~$l$.

\begin{definition}[Locks Held]
\label{def:lh}
  The \emph{lock set for an event $e\in T$ under all schedules} is defined by
  $\GeneralLHeld{T}{e} =
  \{ (l,t) \mid \exists a, r \in T \ \mbox{such that} \
        a = (t,\lockE{l}) \ \mbox{and} \
        e \in \GeneralCSect{T}{a,r} \}$.
\end{definition}

Our definitions correct earlier definitions
as we no longer impose the requirement
that $e$ belongs to the same thread as~$a$ and~$r$.
This also means that events from distinct threads can held the same lock.
Therefore, lock sets contain pairs $(l,t)$
where~$t$ refers to the thread  that acquired and released lock~$l$.

\cref{fig:ex-cs} shows some example traces.
For brevity, we omit request events as they do not matter here.
For each trace we consider critical sections
that are formed by acquiring and releasing lock~$\LKA$.

\begin{example}
\label{ex:lw-but-not-to}
  For trace~\traceRef{TraceCSLW} we find
  that events~$e_2$,...,$e_9$ are in critical
  section $\GeneralCSect{\traceRef{TraceCSLW}}{e_1,e_{10}}$.
  For example, $e_4 \in \GeneralCSect{\traceRef{TraceCSLW}}{e_1,e_{10}}$
  because (1) due to the last-write dependency between~$e_2$ and $e_4$ we have
  that $e_4$ must be preceded by acquire~$e_1$, and
  (2) for similar reasons we find that
  release~$e_{10}$ must be preceded by~$e_4$.
  Hence, $\{ (\LKA, \thread{1}) \} = \GeneralLHeld{\traceRef{TraceCSLW}}{e_4}$
  and also $\{ (\LKA, \thread{1}) \} = \GeneralLHeld{\traceRef{TraceCSLW}}{e_5}$.
\end{example}

\begin{example}
  In case of \traceRef{TraceCSRO} we find that
  $e_5 \not \in \GeneralCSect{\traceRef{TraceCSRO}}{e_3,e_{10}}$
  because $[e_5] \in \crp(\traceRef{TraceCSRO})$.
  This shows that acquire~$e_3$ does not always precede event~$e_5$
  and therefore~$e_5$ is not protected by lock~$\LKA$.
  However, event $e_6$ is protected by lock~$\LKA$.
  Consider $T' \in \crp(\traceRef{TraceCSRO})$ where $e_6$ in $T'$.
  We find that $e_5 \in T'$ (see \textbf{CRP-PO} in \cref{def:crp}),
  $e_2 \in T'$ (see \textbf{CRP-LW}) and
  $e_1 \in T'$ (see \textbf{CRP-PO}).
  There are two acquires~$e_1$ and $e_5$ on the same lock in~$T'$.
  Due to the last-write dependency between $e_2$ and $e_6$,
  acquire~$e_1$ must occur before acquire~$e_5$.
  To guarantee that $T'$ is well formed the release of the earlier
  acquire must be included. Hence, $e_4 \in T'$.
  But then $e_3 \in T'$ and this shows that
  (1) $e_6$ must be preceded by acquire~$e_3$.
  Further reasoning steps show that (2) release~$e_{10}$ must be preceded
  by~$e_6$. From (1) and (2) we conclude that
  $e_6 \in \GeneralCSect{\traceRef{TraceCSRO}}{e_3,e_{10}}$
  and $\{(\LKA,\thread{1}) \} = \GeneralLHeld{\traceRef{TraceCSRO}}{e_6}$.
\end{example}

\begin{example}
  \label{ex:lh-but-not-ro}
  Consider trace~\traceRef{TraceCSROBYND}.
  Events $e_5, e_6, e_7$ are not protected by lock~$\LKA$
  because $[e_5, e_6, e_7] \in \crp(\traceRef{TraceCSROBYND})$
  and this shows that acquire~$e_3$ does not precede any of these events.
  On the other hand, we find that
  $e_8 \in\GeneralLHeld{\traceRef{TraceCSROBYND}}{e_4}$
  for the following reason.
  For any $T' \in \crp(\traceRef{TraceCSROBYND})$
  where $e_8 \in T'$ we have that $e_1,e_2,e_5,e_6,e_7 \in T'$
  due to conditions \textbf{CRP-PO} and \textbf{CRP-LW}.
  There are two acquires~$e_1$ and $e_5$ on the same lock in~$T'$.
  Acquire~$e_1$ must occur before acquire~$e_5$, as otherwise
  the last-write dependency between $e_2$ and $e_8$ will be violated.
  This also means that we must include the matching release~$e_4$ of acquire~$e_1$
  in~$T'$.
  Then, $e_3 \in T'$ where $e_3$ occurs before~$e_8$.
  Hence, (1) $e_8$ is always preceded by~$e_3$.
  Further reasoning steps show that (2) release~$e_{14}$
  must be preceded by~$e_8$.
  From (1) and (2) we conclude that
  $\{(\LKA,\thread{1}) \} = \GeneralLHeld{\traceRef{TraceCSROBYND}}{e_8}$.
\end{example}

%% MS: omit the ``third'' example
%%
%%
%% \begin{example}
%%   For trace~\traceRef{TraceCSNotProtected} we find that
%%   none of the events in thread~$\thread{2}$ are protected by lock~$\LKA$.
%%   This is the case because $[e_1,e_4,e_5] \in \crp(\traceRef{TraceCSNotProtected})$.
%%   Hence, $\{\} = \GeneralLHeld{\traceRef{TraceCSNotProtected}}{e_4}$.
%%   %% MS: omit subtle point
%%   %% if e2 and e7 part of reordering then e4 must be included (and enclosed).
%%   %% But this does not guarantee that e4 is protected by lock l1.
%% \end{example}

\section{Lock Set-Based Deadlock Prediction}
\label{sec:deadlock}

We revisit a standard method for dynamic deadlock prediction.
For a specific program run represented by some trace~$T$
we calculate all lock (set) dependencies in~$T$ and
check for cycles among them.
A \emph{lock dependency} represents the situation
that a request event holds a certain set of locks.
A \emph{cyclic chain of lock dependencies}
where the underlying lock sets are disjoint
is referred to as a \emph{deadlock pattern}.

Definitions of the above notions occur in many forms in the literature.
Our formulation below largely follows~\cite{conf/pldi/TuncMPV23}.
We generalize these definitions and make them parametric
in terms of the lock set function~$L_T$
that maps the set $\evts(T)$ of events in~$T$
to sets of pairs consisting of lock variables and thread ids.
In this section, the only specific instance of~$L_T$ we consider
is $\GeneralLHeldd{T}$ from \cref{def:lh}.
In later sections, we introduce sound approximations
of $\GeneralLHeldd{T}$ that can be efficiently computed.

\subsection{Deadlock Pattern}

\begin{definition}[Lock Dependency]
  Let
  % $L$ be a lock set function and
  $q = (t, \reqE{l}) \in T$ be a request event
  for lock variable $l$.
  We write $\LDInduced{L_T} (q) = \LD{t}{l}{L_T (q)}$ for the
  \emph{lock dependency of $q$ induced by $L_T$}.
\end{definition}

\begin{example}
  \label{ex:pred-dl}
  Consider trace \traceRef{FNExampleLW} in \cref{fig:ex_mt3b}.
  We find that lock dependencies
  $\LDInduced{\GeneralLHeldd{\traceRef{FNExampleLW}}} (e_4)
   = \LD{\thread{2}}{\LKB}{\GeneralLHeld{\traceRef{FNExampleLW}}{e_4}}$ and
  $\LDInduced{\GeneralLHeldd{\traceRef{FNExampleLW}}} (e_{11})
   = \LD{\thread{3}}{\LKA}{\GeneralLHeld{\traceRef{FNExampleLW}}{e_{11}}}$
   where $\GeneralLHeld{\traceRef{FNExampleLW}}{e_4} =  \{(\LKA,\thread{1})\}$
   and $\GeneralLHeld{\traceRef{FNExampleLW}}{e_{11}} =  \{(\LKB,\thread{2})\}$.
   Notice that $\LKA$ occurs in $\GeneralLHeld{\traceRef{FNExampleLW}}{e_4}$
   and $\LKB$ occurs in $\GeneralLHeld{\traceRef{FNExampleLW}}{e_{11}}$.
   Furthermore, lock sets have no common lock.
   Hence, we say that $\LDInduced{\GeneralLHeldd{\traceRef{FNExampleLW}}} (e_4)$
   and $\LDInduced{\GeneralLHeldd{\traceRef{FNExampleLW}}} (e_{11})$
   form a cyclic chain of lock dependencies (aka deadlock pattern).
\end{example}

To formalize deadlock patterns we introduce some notation.
We write $l \in L_T(e)$ as short-hand for $(l,t) \in L_T(e)$
for some $t$. We use this short-hand when checking for (cyclic) lock dependencies
where we ignore the thread id attached to each lock.
However, when checking that lock sets have no common lock, thread ids matter.
The intersection among lock sets only includes lock variables
with different thread ids,
that is, locks that have been acquired in different threads.
\begin{align*}
  M \indexedcap N & := \{ l \mid (l, s) \in M, (l, t) \in N, s \ne t \}
\end{align*}

\begin{definition}[Deadlock Pattern]
  \label{def:deadlock-pattern}
  Let $n>1$ and $q_1,\dots ,q_n \in T$ be request events
  for lock variables $l_1,...,l_n$, respectively.
  % Let $L$ be a lock set function.
  We say that the set of requests
  $\SET{q_1,\dots ,q_n}$ represents a \emph{deadlock pattern}
  if the following conditions hold:
  \begin{description}
  \item[(DP-Thread)] $\thdOf(q_i) \ne \thdOf(q_j)$, for all $i \ne j$.
  \item[(DP-Cycle)] For every $1 \leq i \leq n$,
                $q_i =\reqE{l_i}$ and $l_i \in L_T(q_{(i \mod n) + 1})$.
  \item[(DP-Guard)] $L_T(q_i) \indexedcap L_T(q_j) = \emptyset$, for all $i \ne j$.

  \end{description}
  We use notation $\DP{L_T}{\SET{q_1,\dots ,q_n}}$
  to indicate that $\SET{q_1,\dots ,q_n}$ satisfies the above conditions.

We define the set of deadlock patterns for trace $T$ induced by $L_T$ as follows:
$\AllDP{L_T} = \{ \SET{q_1,\dots ,q_n} \mid \DP{L_T}{\SET{q_1,\dots ,q_n}} \}$.
\end{definition}

Each deadlock pattern $\DP{L_T}{\SET{q_1,\dots ,q_n}}$ corresponds to
a cyclic chain of lock dependencies $\LDInduced{L_T} (q_1), ..., \LDInduced{L_T} (q_n)$.
A deadlock pattern is predictable if there is reordering of the trace
where execution is stuck.

\subsection{Predictable Deadlocks are Stuck Deadlock Patterns}

Earlier formulations of predictable deadlocks~\cite{conf/pldi/TuncMPV23}
rely on a notion of \emph{enabled} events.
An event $e$ is said to be enabled in some correctly reordered prefix $T'$ if
(1) extending $e$ with $T'$ would break any of the well-formedness conditions, and
(2) $T'$ contains all events $f$ that are in the same thread as $e$ but occur before $e$.
In essence, enabled events are `blocked' acquires that arise in some
cyclic lock dependency chain.
In our formulation of deadlock patterns we use request
instead of acquire events.
Hence, we we introduce the notion of \emph{final} events.

\begin{definition}[Final Event]
  Given trace~$T$,
  we say $e$ is \emph{final} in~$T$ if there is no~$f \in T$ such that
  $e \TrLt[T] f$ and $\thd(e) = \thd(f)$.
\end{definition}

Request events are final and their corresponding acquires are enabled.
Here is our reformulation of predictable deadlocks~\cite{conf/pldi/TuncMPV23}
using final instead of enabled events.

\begin{definition}[Predictable Deadlock]
  \label{def:pred-deadlock}
  Given $\SET{q_1,\dots ,q_n} \in \AllDP{L_T}$,
  we say $\SET{q_1,\dots ,q_n}$ is a \emph{predictable deadlock} in~$T$ if
  there exists \emph{witness}~$T' \in \crp(T)$
  such that each of $q_1,...,q_n$ are final in~$T'$.

  We use notation $\PD{L_T}{\SET{q_1,\dots ,q_n}}$ to indicate
that $\SET{q_1,\dots ,q_n}$ satisfies the above conditions.
\end{definition}
We define the set of predictable deadlocks for trace $T$ induced by $L$ as follows:
$\AllPD{L_T} = \{ \SET{q_1,\dots ,q_n} \mid \PD{L_T}{\SET{q_1,\dots ,q_n}} \}$.

An implicit assumption of the above definition is the fact that the witness trace~$T'$
is \emph{stuck} and cannot be extended with the acquire~$a_i$ that
fulfills the lock request~$q_i$.
The key for this assumption is
condition \textbf{DP-Cycle} (see \cref{def:deadlock-pattern})
which states that each lock $\lock(q_i)$ is held by~$q_{(i \mod n) + 1}$.
Extending~$T'$ with $a_i$ therefore violates lock semantics.

We prove \emph{stuckness} for predictable
deadlocks built via~$\GeneralLHeldd{T}$ (\cref{def:lh}).

%% https://tex.stackexchange.com/questions/114151/how-do-i-reference-in-appendix-a-theorem-given-in-the-body
\begin{restatable}{lemma}{stuckness}
\label{le:stuckness}
  Let $\SET{q_1,\dots ,q_n} \in \AllPD{\GeneralLHeldd{T}}$
  where $T'$ is the witness for $\SET{q_1,\dots ,q_n}$.
  Let $a_i$ be the acquire that fulfills the lock request~$q_i$.
  Then, there is no $T''$ such that $T' \conc T'' \conc [a_i]$ is well-formed.
\end{restatable}

\begin{figure}[t]
  \begin{subfigure}[t]{0.4\textwidth}
    \bda{l}
%% latexTrace $ addLoc ex_mt3
\ba{|l|l|l|l|}
\hline
\traceNew{TraceClassicGuard} & \thread{1} & \thread{2} & \thread{3}\\ \hline
\eventE{1}  & \lockE{\LKA}&&\\
\eventE{2}  & \writeE{\VA}&&\\
\eventE{3}  & &\readE{\VA}&\\
\eventE{4}  & &\reqLockE{\LKB}&\\
\eventE{5}  & &\lockE{\LKB}&\\
\eventE{6}  & &\unlockE{\LKB}&\\
\eventE{7}  & &\writeE{\VB}&\\
\eventE{8}  & \readE{\VB}&&\\
\eventE{9}  & \unlockE{\LKA}&&\\
\eventE{10}  & \writeE{\VC}&&\\
\eventE{11}  & &&\readE{\VC}\\
\eventE{12}  & &&\lockE{\LKB}\\
\eventE{13}  & &&\reqLockE{\LKA}\\
\eventE{14}  & &&\lockE{\LKA}\\
\eventE{15}  & &&\unlockE{\LKA}\\
\eventE{16}  & &&\unlockE{\LKB}\\

\hline \ea{}

\eda
      \caption{Not a predictable deadlock.}
 \label{fig:extra-dp-fp}
    \end{subfigure}
  \qquad
    \begin{subfigure}[t]{0.4\textwidth}
      %%latexTrace $ addLoc ex_mt5
      \bda{lcl}
\ba{|l|l|l|l|}
\hline
\traceNew{FalseNegativeMulti} & \thread{1} & \thread{2} & \thread{3}\\ \hline
\eventE{1}  & \lockE{\LKC}&&\\
\eventE{2}  & \writeE{\VA}&&\\
\eventE{3}  & &\readE{\VA}&\\
\eventE{4}  & &\lockE{\LKA}&\\
\eventE{5}  & &\reqLockE{\LKB}&\\
\eventE{6}  & &\lockE{\LKB}&\\
\eventE{7}  & &\unlockE{\LKB}&\\
\eventE{8}  & &\unlockE{\LKA}&\\
\eventE{9}  & &\writeE{\VB}& \\
\eventE{10}  &  &  &\readE{\VA}\\
\eventE{11}  & &&\lockE{\LKB}\\
\eventE{12}  & &&\reqLockE{\LKA}\\
\eventE{13}  & &&\lockE{\LKA}\\
\eventE{14}  & &&\unlockE{\LKA}\\
\eventE{15}  & &&\unlockE{\LKB}\\
\eventE{16}  & &&\writeE{\VC}\\
\eventE{17}  & \readE{\VB}&&\\
\eventE{18}  & \readE{\VC}&&\\
\eventE{19}  & \unlockE{\LKC}&&\\

\hline \ea{}

\eda
      \caption{Predictable deadlock.}
 \label{fig:multi-fn}
   \end{subfigure}

  \caption{Further deadlock pattern examples.}
  \label{fig:predictable-dl-comparison}
  \label{fig:further-dp-examples}
\end{figure}

\subsection{Examples and Intermediate Summary}

\begin{example}
  For trace \traceRef{FNExampleLW} in \cref{fig:ex_mt3b}
  we find that $\SET{e_4,e_{11}}$ form a deadlock pattern.
  This deadlock pattern is a predictable deadlock as shown by
  the witness~$\traceRef{FNExampleLWCRP}$ in \cref{fig:ex_mt1b}.
  We observe that there is no well-formed extension of trace~$\traceRef{FNExampleLWCRP}$ that includes either~$e_5$ or~$e_{12}$.
\end{example}

\begin{example}
  Trace \traceRef{TraceClassicGuard} in \cref{fig:extra-dp-fp}
  similar to trace \traceRef{FNExampleLW}.
  We again find a cyclic lock dependency represented
  by the deadlock pattern $\SET{e_4, e_{13}}$.
  In this case, there is no predicable deadlock.
  Due to the last-write dependency among $e_{10}$ and $e_{11}$
  there cannot be a witness where $e_4$ and $e_{13}$ are final.
\end{example}

\begin{example}
  Trace \traceRef{FalseNegativeMulti} in \cref{fig:multi-fn} contains
  a predictable deadlock and highlights that intersection
  among lock sets takes into account thread ids.
  We find that
  $\GeneralLHeld{\traceRef{FalseNegativeMulti}}{e_5} = \{ (\LKC,\thread{1}), (\LKA,\thread{2}) \}$
  and
  $\GeneralLHeld{\traceRef{FalseNegativeMulti}}{e_{12}} = \{ (\LKC,\thread{1}), (\LKB,\thread{3}) \}$.
  Both lock sets have a common lock~$\LKC$ but this lock
  is connected to the same thread.
  Hence, we do not treat lock~$\LKC$ as a guard lock
  and ignore this lock when building the intersection
  See \textbf{DP-Guard} in \cref{def:deadlock-pattern}.
  Indeed, we find that deadlock pattern $\SET{e_5, e_{12}}$
  represents a predicable lock as shown by
  the witness $[e_1,e_2,e_3,e_{10},e_4,e_{11},e_5,e_{12}]$.
\end{example}

Based on the corrected lock set definition,
we are able to discover a larger class of predictable deadlocks.
For example, the predictable deadlock in trace
\traceRef{FNExampleLW} in \cref{fig:ex_mt3b}
is not detected by formulations that
employ per-thread lock sets.
We can also state that all predictable deadlocks detected
by using per-thread lock sets are also detected by
our more general method. The formal result (\cref{le:pd-to-lw-ro})
is presented in the upcoming section.
In terms of complexity, there are no changes.
Dynamic deadlock prediction remains NP-hard
as stated in~\cite{conf/pldi/TuncMPV23},
regardless if we assume the per-thread or correct lock set construction.
The same arguments used in~\cite{conf/pldi/TuncMPV23}
for the per-thread lock set case also apply to our more general setting.

While being intractable in theory, earlier works show
that dynamic deadlock prediction works well in practice.
This is due to the fact that the per-thread lock set construction
is easy to compute and an exponential blow-up of cyclic lock dependencies
is rare for actual programs.
Our trace-based lock set construction is more involved because we need
to consider all (possibly exponential) trace reorderings.
The upcoming section shows how to effectively approximate
the trace-based lock set construction via partial order relations.
This then leads to an implementation (see \cref{sec:implementation}) that retains
the efficiency of the per-thread lock set construction
but with increased precision (see \cref{sec:experiments}).

\section{Partial Order-Based Critical Section and Lock Set}
\label{sec:partial-order}

We give a characterization of the notions of a critical section
and lock set in terms of partial order relations
under which the set of predictable deadlocks
as described by \cref{def:pred-deadlock} can be soundly approximated.
We review various partial orders that can be found in the literature
and discuss if they are suitable for our purposes.

\subsection{Must Happen Before Approximation}

\begin{definition}[Partial Order]
Given some trace $T$, a \emph{strict partial order}
is a binary relation among events $\evts(T)$ which is
irreflexive, antisymmetric and transitive.
\end{definition}

We often say \emph{partial order} for short and
use symbol~$\pLt[T]$ to refer to some partial order
for some trace~$T$.

The next definitions
rephrase \cref{def:cs} and \cref{def:lh}
in terms of some partial order~$\pLt[T]$.
To make a clear distinction, we will
prefix labels with `\PO' in the upcoming definitions.

\begin{definition}[\PO-Critical Section]
  \label{def:pcs}
  Let~$\pLt[T]$ be some partial order.
  Suppose there are events
  $a = (t, \lockE{l}), r = (t, \unlockE{l}) \in T$ for
  some lock~$l$.
  We say that $e \in T$ is in the
  \emph{critical section enclosed by acquire $\evtAcc$
    and release $\evtRel$ under partial order~$\pLt[T]$},
  written $e \in \CS{\pLt[T]}{a,r}$
  if the following two conditions are met:

  \begin{description}
  \item[(\PO-CS-Enclosed)] $a \pLt[T] e$ and $e \pLt[T] r$.
\item[(\PO-CS-Match)] There is no release event $r'= (t, \unlockE{l}) \in T$  and
       $T' \in \crp(T)$
     such that $a \TrLt[T'] r'$ and $r' \TrLt[T'] e$.
  \end{description}

\end{definition}

\begin{definition}[\PO-Locks Held]
 \label{def:p-lh}
  The \emph{lock set for an event $e\in T$ under partial order~$\pLt[T]$} is defined by
  $\LH{\pLt[T]}{e} =
  \{ (l,t) \mid \exists a, r \in T \ \mbox{such that} \
    a = (t,\lockE{l}) \ \mbox{and} \
     e \in \CS{\pLt[T]}{a,r}$.
\end{definition}

We are particularly interested in partial orders
that underapproximate the set of all correct reorderings.
By underapproximate we mean that if
two events $e$ and $f$ are ordered such that $e \pLt[T] f$,
then in all reorderings $e$ appears before~$f$.
The following definition formalizes this requirement.

\begin{definition}[Must Happen Before Criteria]
  We say~$\pLt[T]$ satisfies
  the \emph{must happen before} criteria if
  the following condition is met.

  \begin{description}
\item[MHB-Criteria] If $e \pLt[T] f$ for some trace $T$
  then for all $T' \in \crp(T)$ where $f \in T'$ we find that
  $e\in T'$ and ~$e \TrLt[T'] f$.
  \end{description}

\end{definition}

All \emph{underapproximating} partial orders
that satisfy the \textbf{MHB-Criteria}
are sound in the sense that every lock held
discovered by the approximation is an actual lock held.

\begin{restatable}{lemma}{mhbcriteriainclusion}
 \label{le:mhb-criteria-inclusion}
  Let $\pLt[T]$ be some partial order satisfying
  the \textbf{MHB-Criteria}.
  Then, $\LH{\pLt[T]}{e} \subseteq \GeneralLHeld{T}{e}$ for all $e \in T$.
\end{restatable}

Importantly, the set of predictable deadlocks
can be approximated by  partial orders
that satisfy the \textbf{MHB-Criteria}.

\begin{restatable}{lemma}{lepdmhb}
  \label{le:pd-mhb}
    Let $\pLt[T]$ be some partial order satisfying
    the \textbf{MHB-Criteria}.
    Then, $\AllPD{\LHH{\pLt[T]}} \subseteq \AllPD{\GeneralLHeldd{T}}$.
\end{restatable}

It immediately follows that the stuckness guarantee (\cref{le:stuckness})
also holds for predicable deadlocks built via $\LHH{\pLt[T]}$.

\begin{lemma}
\label{le:stuckness-p}
  Let $\SET{q_1,\dots ,q_n} \in \AllPD{\LHH{\pLt[T]}}$
  where $T'$ is the witness for $\SET{q_1,\dots ,q_n}$.
  Let $a_i$ be the acquire that fulfills the lock request~$q_i$.
  Then, there is no $T''$ such that $T' \conc T'' \conc [a_i]$ is well-formed.
\end{lemma}

\begin{figure}[t]

  \begin{subfigure}[t]{0.42\textwidth}
    \bda{l}
%%latexTrace $ addLoc ex_mt8
\ba{|l|l|l|l|}
\hline
\traceNew{HBTrace} & \thread{1} & \thread{2} & \thread{3}\\ \hline
\eventE{1}  & \lockE{\LKA}&&\\
\eventE{2}  & \lockE{\LKC}&&\\
\eventE{3}  & \unlockE{\LKC}&&\\
\eventE{4}  & &\lockE{\LKC}&\\
\eventE{5}  & &\unlockE{\LKC}&\\
\eventE{6}  & &\reqLockE{\LKB}&\\
\eventE{7}  & &\lockE{\LKB}&\\
\eventE{8}  & &\unlockE{\LKB}&\\
\eventE{9}  & &\lockE{\LKC}&\\
\eventE{10}  & &\unlockE{\LKC}&\\
\eventE{11}  & \lockE{\LKC}&&\\
\eventE{12}  & \unlockE{\LKC}&&\\
\eventE{13}  & \unlockE{\LKA}&&\\
\eventE{14}  & &&\lockE{\LKB}\\
\eventE{15}  & &&\reqLockE{\LKA}\\
\eventE{16}  & &&\lockE{\LKA}\\
\eventE{17}  & &&\unlockE{\LKA}\\
\eventE{18}  & &&\unlockE{\LKB}\\

\hline \ea{}
\eda
      \caption{$e_1 \hbLt[\traceRef{HBTrace}] e_6 \hbLt[\traceRef{HBTrace}] e_{13}.$}
 \label{fig:hb-not-mhb}
\end{subfigure}
  \qquad
  \begin{subfigure}[t]{0.42\textwidth}
    \bda{l}
\ba{|l|l|l|l|}
\hline
\traceNew{HBTraceWitness} & \thread{1} & \thread{2} & \thread{3}\\ \hline
\eventE{1}  & \lockE{\LKA}&&\\
\eventE{14}  & &&\lockE{\LKB}\\
\eventE{2}  & \lockE{\LKC}&&\\
\eventE{3}  & \unlockE{\LKC}&&\\
\eventE{4}  & &\lockE{\LKC}&\\
\eventE{5}  & &\unlockE{\LKC}&\\
\eventE{6}  & &\reqLockE{\LKB}&\\
\eventE{15}  & &&\reqLockE{\LKA}\\
 \eventE{11}  & \lockE{\LKC}&&\\
 \eventE{12}  & \unlockE{\LKC}&&\\
 \eventE{13}  & \unlockE{\LKA}&&\\
 \eventE{16}  & &&\lockE{\LKA}\\
 \eventE{17}  & &&\unlockE{\LKA}\\
 \eventE{18}  & &&\unlockE{\LKB}\\
 \eventE{7}  & &\lockE{\LKB}&\\
 \eventE{8}  & &\unlockE{\LKB}&\\
 \eventE{9}  & &\lockE{\LKC}&\\
 \eventE{10}  & &\unlockE{\LKC}&\\

\hline \ea{}

\eda
\caption{Witness $[e_1,e_{14},e_2,e_3,e_4,e_5,e_6,e_{15}]$ is not stuck
as shown by the extension $[e_{11},e_{12},e_{13},e_{16},e_{17},e_{18},e_7,e_8,e_9,e_{10}]$.}
 \label{fig:witness-not-stuck}
  \end{subfigure}
  \caption{\HB\ does not satisfy the \textbf{MHB-criteria}. Stuckness no longer holds.}
\label{fig:hb-not-stuck}
\end{figure}

\subsection{Overapproximating Partial Order Relations}

The literature offers a wealth of partial order relations~\cite{Smaragdakis:2012:SPR:2103621.2103702,conf/pldi/KiniMV17,Mathur:2018:HFR:3288538:3276515,10.1145/3360605,conf/mplr/SulzmannS20}.
Many of these partial orders are not suitable for our purposes because
they violate the \textbf{MHB-criteria}.
As a representative example, we
consider Lamport’s happens-before (HB) relation~\cite{lamport1978time}.

\begin{definition}[Happens-Before Order (\HB) \cite{lamport1978time}]
  \label{def:hb-order}
  Given trace $T$, we define the \emph{happens-before-order (\HB)}
  as the smallest strict partial order on $\evts(T)$
  such that
  \begin{description}
  \item[\nf{(a)}] $e \hbLt[T] f$ if
  $e \TrLt[T] f$ and $\thd(e) = \thd(f)$ ,
  \item[\nf{(b)}] $r \hbLt[T] a$ if $r=(\_,\unlockE{l})$ and $a=(\_,\lockE{l})$
    and $r \TrLt[T] a$.
   \end{description}
\end{definition}

\HB\ does not satisfy the \textbf{MHB-criteria}
and therefore~\cref{le:stuckness} does not hold for $\AllPD{\LHH{\hbLt[T]}}$.

\begin{example}
Consider \traceRef{HBTrace} in \cref{fig:hb-not-mhb}.
We find that $e_1 \hbLt[\traceRef{HBTrace}] e_6 \hbLt[\traceRef{HBTrace}] e_{13}$
and therefore $l_1 \in \HBLHeld{\traceRef{HBTrace}}{e_6}$.
We conclude that $\SET{e_6, e_{15}} \in \AllDP{\LHH{\hbLt[T]}}$
is a deadlock pattern.
In fact, $\SET{e_6, e_{15}} \in \AllPD{\LHH{\hbLt[T]}}$
is also a predictable deadlock as shown
by the witness~$[e_1,e_{14},e_2,e_3,e_4,e_5,e_6,e_{15}]$.
See \cref{fig:witness-not-stuck}.
However, this witness is not stuck as there is an extension
where we can include the acquire~$e_7$ to fulfill the request~$e_6$
and the acquire~$e_{16}$ to fulfill the request~$e_{15}$.
The reason why this is possible is that \HB\ does not satisfy
the \textbf{MHB-criteria}.
For example, we can place the events $[e_{11},e_{12},e_{13}]$
above the acquire $e_7$.
We can even place events $[e_{11},e_{12},e_{13}]$
above the request~$e_6$.
Although $l_1 \in \LH{\hbLt[\traceRef{HBTrace}]}{e_6}$ this does not
guarantee that request~$e_6$ is protected by lock~$l_1$.
\end{example}

Similar observations apply to the
SHB~\cite{Mathur:2018:HFR:3288538:3276515},
WCP~\cite{conf/pldi/KiniMV17}
and
SDP~\cite{10.1145/3360605}
partial orders.

\subsection{Underapproximating Partial Order Relations}

Fortunately, there are a number of (underapproximating) partial orders
that satisfy the \textbf{MHB-criteria}.
In the following, we consider three examples.
All three examples are well-known and
can be found in one form or another in the literature.
Their use for the purpose of lock set-based dynamic deadlock prediction
is novel, however.

The most simple partial order is the one that
only orders events in the same thread based on their textual position.

\begin{definition}[Thread-Order (\TO)]
  \label{def:thread-order}
  Given trace $T$, we define the \emph{thread-order (\TO)}
  as the smallest strict partial order~$\toLt[T]$ on $\evts(T)$
  such that $e \toLt[T] f$ if
  $e \TrLt[T] f$ and $\thd(e) = \thd(f)$.
\end{definition}

We strengthen thread-order by including last-write dependencies.

\begin{definition}[Last-Write-Order (\LW)]
  \label{def:last-write-order}
  Given trace $T$, we define the \emph{last-write-order (\LW)}
  as the smallest strict partial order~$\lwLt[T]$ on $\evts(T)$
  such that
  \begin{description}
   \item[\nf{(a)}] $\toLt[T] \subseteq \lwLt[T]$,
   \item[\nf{(b)}] $e \lwLt[T] f$ if $e = \lw[T](f)$.
  \end{description}
\end{definition}
Recall that for read~$f$, function $\lw[T](f)$ yields
the write on the same shared variable that immediately precedes
$f$ in trace $T$.

The final partial order imposes an ordering among critical sections.
WCP~\cite{conf/pldi/KiniMV17} and SDP~\cite{10.1145/3360605} order critical sections
but are not suitable because they overapproximate and
violate the \textbf{MHB-criteria}.
We follow WDP~\cite{10.1145/3360605} and PWR~\cite{conf/mplr/SulzmannS20}
and only order critical sections if absolutely necessary.
We only consider `conflicts' if they involve last write dependencies.

%% MS:
%% a and r don't need to be part of their own critical section.
%% We consider here underapproximations. So, it's fine to consider a variant of PWR.
\begin{definition}[Release-Order (\RO)]
  \label{def:release-order}
  Given trace $T$, we define the \emph{release-order (\RO)}
  as the smallest strict partial order~$\roLt[T]$ on $\evts(T)$
  such that
  \begin{description}
   \item[\nf{(a)}]
  $\lwLt[T] \subseteq \roLt[T]$,
     \item[\nf{(b)}] $r \roLt[T] f$ if
  $e \in \CS{\toLt[T]}{a,r}$ and
       $f\in \CS{\toLt[T]}{a',r'}$,
         $a=(\_,\lockE{l})$, $a'=(\_,\lockE{l})$
       and $e \lwLt[T] f$.
  \end{description}
\end{definition}
Examples that show the above partial orders in action are coming up shortly.
We first state some important properties.

$\LHH{\toLt[T]}$ corresponds to the existing (standard) lock set construction.
All three partial orders $\toLTSym$, $\lwLTSym$ and $\roLTSym$
enjoy the \textbf{MHB-Criteria} and yield an (increasingly)
improved lock set construction which leads to increased
precision of lock set-based dynamic deadlock prediction.

\begin{restatable}{lemma}{leromhb}
  \label{le:ro-mhb}
(a) $\toLTSym \subseteq \lwLTSym \subseteq \roLTSym$.
(b) $\toLTSym$, $\lwLTSym$ and $\roLTSym$ satisfy the \textbf{MHB-Criteria}.
\end{restatable}

Via \cref{le:mhb-criteria-inclusion}, \cref{le:pd-mhb}
and \cref{le:ro-mhb} we arrive at the following result.

\begin{lemma}
\label{le:pd-to-lw-ro}
(a) $\LH{\toLt[T]}{e} \subseteq \LH{\lwLt[T]}{e} \subseteq \LH{\roLt[T]}{e} \subseteq \GeneralLHeld{T}{e}$ for all $e \in T$.
(b) $\AllPD{\LHH{\toLt[T]}} \subseteq \AllPD{\LHH{\lwLt[T]}} \subseteq \AllPD{\LHH{\roLt[T]}} \subseteq \AllPD{\GeneralLHeldd{T}}$.
\end{lemma}

The inclusion (a) in \cref{le:pd-to-lw-ro}
is strict as shown by the following examples
where we revisit the traces in \cref{fig:ex-cs}.

\begin{example}
  Consider trace~\traceRef{TraceCSLW}.
  Based on the arguments in~\cref{ex:lw-but-not-to}
  we find that   $\{ (\LKA, \thread{1}) \} = \LH{\lwLt[\traceRef{TraceCSLW}]}{e_5}$
  but clearly $\{  \} = \LH{\toLt[\traceRef{TraceCSLW}]}{e_5}$.
\end{example}

\begin{example}
  For \traceRef{TraceCSRO} we find
  that (1) $e_6 \lwLt[\traceRef{TraceCSRO}] e_9$
  due to thread-order and the last-write dependency between $e_8$ and $e_9$.
  However, $e_3 \not \lwLt[\traceRef{TraceCSRO}] e_6$ and
  therefore $\{\} = \LH{\lwLt[T]}{e_6}$.
  Because partial order \RO\ orders conflicting critical sections we
  find that (2) $e_3 \roLt[\traceRef{TraceCSRO}] e_6$.
  From (1) and (2) we conclude that $\{(\LKA,\thread{1})\} = \LH{\roLt[\traceRef{TraceCSRO}]}{e_6}$.
  (2) holds for the following reason.
  We find that $e_2 \in \CS{\toLt[\traceRef{TraceCSRO}]}{e_1,e_4}$
  and $e_6 \in \CS{\toLt[\traceRef{TraceCSRO}]}{e_5,e_7}$.
  Furthermore, $e_2 \lwLt[\traceRef{TraceCSRO}] e_6$
  and thus we obtain (2) by application of case (b) in~\cref{def:release-order}.
\end{example}

\begin{example}
  Consider trace~\traceRef{TraceCSROBYND}.
  The reasoning in \cref{ex:lh-but-not-ro}
  shows that $\{(\LKA,\thread{1}) \} = \GeneralLHeld{\traceRef{TraceCSROBYND}}{e_8}$
  but $\{ \} =  \LH{\roLt[\traceRef{TraceCSROBYND}]}{e_8}$
  for the following reason.
  We have that $e_2 \roLt[\traceRef{TraceCSROBYND}] e_8$
  and $e_2 \in \CS{\toLt[\traceRef{TraceCSROBYND}]}{e_1,e_3}$
  but $e_8$ is not part of any `standard' critical section for lock~$l_1$.
  For example, $\CS{\toLt[\traceRef{TraceCSROBYND}]}{e_5,e_{11}}$
  includes $e_6$, $e_{10}$ but not $e_7$, $e_8$ and $e_9$.
\end{example}

%%To summarize.
%% There is always the trade-off between precision and performance.
The upcoming section shows that the partial order-based lock set constructions
can be computed efficiently and easily integrated into an
existing lock set-based deadlock predictor.
Compared to trace-based lock set characterization there is a loss in precision.
The inclusions (a) and (b) in \cref{le:pd-to-lw-ro} are strict.
Case (b) can be shown by extending the traces in \cref{le:pd-to-lw-ro}.

%% MS: omit, less is more, might raise too many questions.

%% The subtle point in the above example is that
%% partial order $\roLt[T]$ as defined in \cref{def:release-order}
%% only checks for conflicts among standard critical sections $\CS{\toLt[T]}$ based
%% on the $\toLt[T]$ partial order.
%% Checking if an event is part of a standard critical section $\CS{\toLt[T]}$

%%
%% We could increase the precision of $\roLt[T]$ by for example
%% employing $\CS{\lwLt[T]}$ instead of $\CS{\toLt[T]}$
%% in case (b) of~\cref{def:release-order}.
%% Then, $e_8 \in \CS{\lwLt[\traceRef{TraceCSROBYND}]}{e_5,e_{11}}$ and
%% the thus extended \RO\ partial order yields that $e_8$ is protected by lock~$l_1$.

\section{Implementation}
\label{sec:implementation}

We integrate the partial order-based lock set constructions
into the state of the art deadlock predictor~\SPDOffline~\cite{conf/pldi/TuncMPV23}.
\SPDOffline operates \emph{offline} by taking as input a trace to compute
a set of \emph{sync-preserving} deadlocks.
Sync-preserving deadlocks are a subclass of the set of predictable deadlocks
where in the witness trace the relative order of acquires on the same lock
is the same as in the input trace.
All of our results also apply to sync-preserving deadlocks
and all of the deadlocking examples in the main text of this paper
are sync-preserving.

\SPDOffline\ has two phases:
\begin{description}
  \item[(Phase 1)] Computation of deadlock patterns.
  \item[(Phase 2)] Checking if a particular instance of a deadlock pattern
represented as a set $\SET{q_1,...,q_n}$ of requests has
sync-preserving witness where $q_1,...,q_n$ are final.
\end{description}

Both phases are implemented by~\cref{alg:lock-deps} and~\cref{alg:dp-syncp}.
\cref{alg:lock-deps} computes deadlock patterns
where we make use of the partial order-based lock set constructions.
\cref{alg:dp-syncp} follows closely the original description~\citep[Algorithm 2]{conf/pldi/TuncMPV23}.
We first consider the computation of deadlock patterns (\cref{sec:general-lockset-algo})
before discussing the integration into \SPDOffline (\cref{sec:spd-ext}).

\subsection{Computation of Deadlock Patterns} %%{Computing Partial Order-Based Lock Sets}
\label{sec:general-lockset-algo}

\begin{algorithm*}[t!]
  \caption{Computation of \LW\ deadlock patterns.}
  \label{alg:lock-deps}

  {\small
    \begin{algorithmic}[1]
      \Function{computeLWDeadlockPattern}{$T$}
      \label{ln:cLDs}
      \State $\GlobalLHeldT = \emptyset$
        \Comment{Global lock set}
        \label{ln:global-lh}
        \State $\forall t \colon \StdLHeldT(t) = \emptyset$
        \Comment{Standard lock set of thread $t$}
        \label{ln:lh}
        \State $\forall t \colon \threadVC{t} = [\bar{0}]; \incC{\threadVC{t}}{t}$
        \Comment{Vector clock $\threadVC{t}$ of thread $t$}
        \label{ln:thvc}
        \State $\forall x \colon \lastWriteVC{x} = [\bar{0}]$
        \Comment{Vector clocks $\lastWriteVC{x}$ of most recent $\writeE{x}$}
        \label{ln:lwlr}
        \State $\forall l \colon \acqVC{l} = [\bar{0}]$
        \Comment{Vector clock $\acqVC{l}$ of most recent $\acqE{l}$}
        \label{ln:acqv}
        \State $\LDMapSym = \emptyset$
        \Comment{Map  with keys $(t,l,ls)$, value equals a list of vector clocks}
        \label{ln:ld}
        \State $\PList = \emptyset$
        \Comment{List of lock dependencies (t,l,ls,gs) with potential guards $gs$}
        \label{ln:pl}
        \ForDo {$e$ in $T$} {\Call{process}{$e$}}
        %% MS: omit
        %% practical issue but we assume that there's a release for each acquire
%%         \For  {$(t, l, V, ls, gs) \in \PList$}
%%         \If {$ls \not= \emptyset$}
%%         \State $\LDMap{t}{l}{ls}.push(V)$
%%         \label{ln:not-verified-cases}
%%         \EndIf
%%         \EndFor
              \State $DP = \{ (\LDMap{\thread{1}}{l_1}{ls_1}, \ldots, \LDMap{\thread{n}}{l_n}{ls_n}) \mid \forall i \ne j \colon l_i \in ls_{(i \mod n) + 1} \wedge ls_i \indexedcap ls_j = \emptyset \}$
      \label{ln:cyclic-chain}
        \State \Return DP
      \EndFunction
      \algstore{cld}
    \end{algorithmic}

    \begin{minipage}[t]{.52\textwidth}
      \begin{algorithmic}[1]
        \algrestore{cld}
        \Procedure{process}{$(t,acq(l))$}
        \label{ln:procAcq}
        \State $gs = \emptyset$
        \For {$l' \in \GlobalLHeldT \setminus \StdLHeldT(t)$}
        \label{ln:acq-potential-guards}
        \If {$\acqVC{l'} < \threadVC{t}$}
        \label{ln:acq-potential-guards-need-to-be-verified}
        \State $gs = gs \cup \{l'\}$
        \EndIf
        \EndFor
          \If {$gs = \emptyset \wedge \StdLHeldT(t) \not= \emptyset$}
          \State $\LDMap{t}{l}{\StdLHeldT(t)}.push(\threadVC{t})$
          \label{ln:acq-standard-lock-dep}
          \EndIf
          \If {$gs \not= \emptyset$}
          \State $\PList.push((t,l,\threadVC{t},\StdLHeldT(t),gs))$
          \label{ln:acq-possible-guards}
          \EndIf
          \State $\StdLHeldT(t) = \StdLHeldT(t) \cup \{ l \}$
          \State $\GlobalLHeldT = \GlobalLHeldT \cup \{ l \}$
          \label{ln:acq-lh-updates}
          %% The above represents the ``pre'' vector clock = request
          \State $\acqVC{l} = \threadVC{t}$
          \label{ln:acq-store}
          \State $\incC{\threadVC{t}}{t}$
          \label{ln:acq-vc-inc}
        \EndProcedure
        \algstore{eacq}
      \end{algorithmic}
    \end{minipage}%
    \hfill%
    \begin{minipage}[t]{.44\textwidth}

      \begin{algorithmic}[1]
        \algrestore{eacq}
        \Procedure{process}{$(t,rel(l))$}
        \label{ln:procRel}
        \For {$p@(t', l', V, ls, gs) \in \PList$}
        \label{ln:rel-check-for-guards}
        \If {$l \in gs \wedge V < \threadVC{t}$}
        \State $gs = gs \setminus \{l\}; ls = ls \cup \{l\}$
        \label{ln:rel-guard-verified}
        \EndIf
        \If {$gs = \emptyset$}
        \State $\LDMap{t'}{l'}{ls}.push(V)$
        \State $\PList.erase(p)$
        \EndIf
        \EndFor
          \State $\StdLHeldT(t) = \StdLHeldT(t) \setminus \{ l \}$
          \State $\GlobalLHeldT = \GlobalLHeldT \setminus \{ l \}$
          \label{ln:rel-lh-updates}
          \State $\incC{\threadVC{t}}{t}$
          \label{ln:rel-vc-inc}
        \EndProcedure
        \algstore{erel}
      \end{algorithmic}

    \end{minipage}%

    \vspace{.3cm}

    \begin{minipage}[t]{.46\textwidth}
      \begin{algorithmic}[1]
        \algrestore{erel}
        \Procedure{process}{$(t,wr(x))$}
          \label{ln:procWr}
          \State $\lastWriteVC{x} = \threadVC{t}$
          \label{ln:lw-store}
          \State $\incC{\threadVC{t}}{t}$
          \label{ln:wr-vc-inc}
        \EndProcedure
        \algstore{ewr}
      \end{algorithmic}
       \end{minipage}
    \hfill%
        \begin{minipage}[t]{.46\textwidth}
      \begin{algorithmic}[1]
        \algrestore{ewr}
        \Procedure{process}{$(t,rd(x))$}
          \label{ln:procRd}
          \State $\threadVC{t} = \threadVC{t} \sqcup \lastWriteVC{x}$
          \State $\incC{\threadVC{t}}{t}$
          \label{ln:rd-vc-inc}
        \EndProcedure
      \end{algorithmic}

    \end{minipage}
  }
\end{algorithm*}

\algdef{SE}[DOWHILE]{Do}{DoWhile}{\algorithmicdo}[1]{\algorithmicwhile\ #1}%

\begin{algorithm*}[t]
  \caption{Checking for sync-preserving witness.}
  \label{alg:dp-syncp}

  {\small

    \begin{algorithmic}[1]
      \Function{computeSPDlks}{$T$}
      \label{ln:computeSPDlks}
      \State $DP = \Call{computeLWDeadlockPattern}{T}$
      \label{ln:call-computeDP}
      \State $Dlks = \emptyset$
      \For{$(F_1,\ldots,F_n) \in DP$}
      \If{$\Call{checkSPD}{t_1,F_1,\ldots,t_n,F_n} = (\OK,(V_1,\ldots,V_n))$}
      \label{ln:call-checkSPD}
      \State $Dlks = Dlks \cup \{ \{ V_1,\ldots,V_n \} \}$
      \EndIf
      \EndFor
      \State \Return $Dlks$
      \EndFunction
      \algstore{computeSPDlks}
    \end{algorithmic}

    \begin{minipage}[t]{.56\textwidth}
      \begin{algorithmic}[1]
        \algrestore{computeSPDlks}
        \Function{checkSPD}{$t_1,F_1,\ldots,t_n,F_n$}
        \label{ln:check-spd}
        \State $k_i = 0$ for $i=1,\ldots,n$
        \label{ln:index-init}
        \While{$\bigwedge_{i=1}^n k_i < F_i.size$}
        \label{ln:iterate-through-all}
        \State $V_i = F_i[k_i]$ for $i=1,\ldots,n$
        \label{ln:index-access-vc}
        \State $V'_i = \incC{F_i[k_i]}{t_i}$ for $i=1,\ldots,n$
        \label{ln:vc-of-acquire}
        \State $V = \Call{compSPClosure}{V_1,...,V_n}$
        \label{ln:spclosure}
        \If{$\forall i. V < V'_i$}
        \label{ln:dp-trw-okay}
        \State \Return $(\OK,(V_1,\ldots,V_n))$
        \EndIf
        \State $V = V_1 \sqcup \ldots \sqcup V_n$
        \label{ln:sync-candidates}
        \State $k_i = \Call{next}{V,F_i,k_i}$ for $i=1,\ldots,n$
        \label{ln:next-candidates}
        \EndWhile
        \EndFunction
        \State \Return $(\FAIL,\dontCare)$
        \algstore{checkSPD}
      \end{algorithmic}
    \end{minipage}%
    \hfill%
    \begin{minipage}[t]{.4\textwidth}

      \begin{algorithmic}[1]
        \algrestore{checkSPD}
        \Function{next}{$V,F,i$}
        \Do
        \State $(\dontCare,V',\dontCare) = F[i]$
        \If{$\neg (V' < V)$}
        \label{ln:next-entry}
        \State \Return $i$
        \EndIf
        \State $i = i + 1$
        \DoWhile{$i < F.size$}
        \State \Return $i$
        \EndFunction
      \end{algorithmic}

    \end{minipage}
  }
\end{algorithm*}

For efficient implementation of partial orders we make use of vector clocks~\cite{Fidge:1991:PAT:646210.683620,Mattern89virtualtime}.
We repeat some standard definitions for vector clocks.

\begin{definition}[Vector Clocks]
  A \emph{vector clock} $V$ is a list of time stamps of the form
  % \[
  $
    % V ::=
    [i_1,\ldots,i_n]
  $.
  % \]
  We assume vector clocks are of a fixed size $n$.
  \emph{Time stamps} are natural numbers, and a time stamp at position $j$ corresponds to the thread with identifier $\thread{j}$.

  We define
  % \[
  $
    [i_1,\ldots, i_n] \sqcup [j_1,\ldots,j_n] = [\maxN{i_1}{j_1},\ldots,\maxN{i_n}{j_n}]
  $
  % \]
  to \emph{synchronize} two vector clocks, taking pointwise-maximal time stamps.
  We write $\incC{V}{j}$ to denote incrementing the vector clock $V$ at position $j$ by one,
  and $\accVC{V}{j}$ to retrieve the time stamp at position $j$.

  We write $V < V'$ if $\forall k \colon \accVC{V}{k} \leq \accVC{V'}{k} \wedge \exists k \colon \accVC{V}{k} < \accVC{V'}{k}$,
  and $V \VCConc V'$ if $V \not< V'$ and $V' \not< V$.
\end{definition}

\cref{alg:lock-deps} computes deadlock patterns based on the $\lwLt[T]$ partial order.
We refer to such deadlock patterns as \emph{\LW\ deadlock patterns}.
In case of deadlock patterns based on the $\roLt[T]$ partial order
the differences only concern the computation of vector clocks.
Something that has been studied extensively~\cite{Smaragdakis:2012:SPR:2103621.2103702,conf/pldi/KiniMV17,10.1145/3360605,conf/mplr/SulzmannS20}.
Hence, we focus on the construction \LW\ deadlock patterns.

Function \textsc{computeLWDeadlockPattern} takes as input
a trace and yields a set of $\LW$ deadlock patterns.
We maintain several state variables when processing events.
For every thread $t$, we have a vector clock~$\threadVC{t}$.
Vector clocks $\lastWriteVC{x}$ represent the last write event on variable~$x$.
Initially, all time stamps are set to zero (\cref{ln:thvc,ln:lwlr,ln:acqv}; $\bar{0}$ denotes a sequence of zeros), and $\threadVC{t}$ is set to 1 at position $t$ (\cref{ln:thvc}).
We use these vector clocks to compute the $\lwLt[T]$ partial order.
See the processing of read (\cref{ln:procRd}) and write (\cref{ln:procWr}) events.

For brevity, we omit explicit request events.
We assume that the vector clock before processing acquire events
implicitly represents requests.

When processing acquire and release events, we maintain two kind of lock sets.
We maintain the standard lock set $\StdLHeldT$ (corresponds to $\LHH{\toLt[T]}$)
and the `global' lock set $\GlobalLHeldT$.
The set $\GlobalLHeldT$ overapproximates the set $\LHH{\lwLt[T]}$.
When processing an acquire, we check for all locks~$l'$
that are not in $\StdLHeldT$ but could be potentially
in $\LHH{\lwLt[T]}$.
Lock~$l'$ is a potential candidate, if the operation that has acquired lock~$l'$
is in a $\LW$ ordering relation with the current acquire of lock~$l$.
To perform this check, we record the vector clock of each
acquire operation (\cref{ln:acq-store}).
We then use this information to check
if $l'$ is a potential candidate in $\LHH{\lwLt[T]}$ (\cref{ln:acq-potential-guards-need-to-be-verified}).
We collect all potential candidates in~$gs$.

If $gs$ is empty and the standard lock set is non-empty,
we can immediately record a (standard) lock dependency (\cref{ln:acq-standard-lock-dep}).
For fast access, we use a map to record lock dependencies.
There may be several acquires that share the same keys.
We store all acquires in a vector using their vector clock as a representative.

If $gs$ is non-empty we add the \emph{potentially guarded}
lock dependency
$
(t,l,\threadVC{t},\StdLHeldT(t),gs)
$
to the list~\PList.
By \emph{potentially guarded} we mean that when processing
the release of any of the potential candidates in $gs$,
we need to verify if the release is in a \LW\ ordering relation
with the current acquire.

When processing a release operation on lock~$l$,
we consult $\PList$ to check if $l$ appears as a potential guard in any $gs$.
For convenience, we make use of set notation when processing list elements.
We use pattern matching notation to refer to element $p$
of shape $(t', l', V, ls, gs)$ (\cref{ln:rel-check-for-guards}).
In our actual implementation, we do not naively iterate over the entire list
but selectively check if a release could be a potential guard.

If the current release is a potential guard, we need to check
that this release is in a $\lwLt[T]$ ordering relation with the prior acquire
that is potentially protected by lock~$l$.
This check can be again carried out via vector clocks.
For each lock dependency $(t', l', ls)$, we also record the vector clock $V$
of the acquire at the time we computed dependency $(t', l', ls)$
and the potential guards $gs$.
Hence, we only need to check if $V < \threadVC{t}$.
If yes, the potential guard could be verified as
an element in $\LHH{\lwLt[T]}$ (\cref{ln:rel-guard-verified}).

%% MS: omit
%%  arises in practice, but we assume there's a release for each acquire
%%
%% Once we are done with processing all events, there might be still some
%% elements left in $\PList$. This can happen in case of an acquire without
%% a matching release.
%% In such a case, all potential guards locks will be eliminted.
%% If the standard lock set is non-empty, we make sure
%% to add these dependencies to the map $\LDMap{}{}{}$ (\cref{ln:not-verified-cases}).

Finally, we compute deadlock patterns satisfying
conditions \textbf{DP-Cycle} and \textbf{DP-Guard} (\cref{ln:cyclic-chain}).
We assume thread-indexed lock sets.
Each lock has an attribute to identify the thread from where
the lock results from. For brevity, we omit details.

For efficiency reasons, our representation of deadlock patterns
differs from~\cref{def:deadlock-pattern}.
In~\cref{def:deadlock-pattern}, we represent
a deadlock pattern as a set $\SET{q_1,\dots ,q_n}$ of requests.
In~\cref{sec:general-lockset-algo}, each request $q_i$ is represented
by a vector clock.
We use the vector clock $\threadVC{t}$
before processing of the acquire $(t,acq(l))$ to represent the request.
All (request) vector clocks that are connected to the
same lock dependency $(l,t,ls)$ are stored in a map $\LDMapSym$
where we use $l$,$t$ and $ls$ as keys.
Each value in the map is a list of (request) vector clocks.
Hence, a deadlock pattern is a tuple $(F_1,...,F_n)$
where each $F_i$ refers to a list of vector clocks.
This is referred to as an \emph{abstract} deadlock pattern~\cite{conf/pldi/TuncMPV23}
and simplifies checking for true positives
among instance $(V_1,...,V_n)$ where $V_i$ is an element in $F_i$.

\subsection{\SPDOffline\ Integration}
\label{sec:spd-ext}

\cref{alg:dp-syncp} describes the second phase of \SPDOffline.
In this version we use \LW\ deadlock patterns.
But it is easy to use standard or \RO\ deadlock patterns instead
by adjusting the vector clock computations in the first phase.
We then verify if each abstract deadlock pattern entails
a sync-preserving deadlock (\cref{ln:check-spd}).
For convenience we assume that vector clocks in the list $F_i$ are in thread $\thread{i}$.
Each instance of an abstract deadlock pattern
is verified by procedure \textsc{compSPClosure}~\citep[Algorithm 1]{conf/pldi/TuncMPV23}.
Via the \textsc{next} function we systematically explore all instances.

Importantly, the soundness guarantee of \SPDOffline\  also holds for our extended version thanks to~\cref{le:stuckness}.
Time complexity of \SPDOffline\ is also not affected based on the following reasoning.
Let $N$ be the number of events, $K$ be the number of threads and $C$ be the number of critical sections (acquire-release pairs).
Then, computation of (additional) lock dependencies takes time $O(N \cdot K \cdot C)$ because in each step ($1,...,N$)
in \cref{sec:general-lockset-algo} we might have $O(C)$ lock dependency candidates to consider and each vector clock operation takes time $O(K)$.
Generally, $K$ and $C$ are treated like constants. Thus, computation of lock dependencies takes time $O(N)$.
Under this assumption we can also argue that computation of vector clocks to compute partial orders $\toLt[T]$, $\lwLt[T]$ and $\roLt[T]$ takes time $O(n)$.
See~\cite{conf/pldi/KiniMV17} where the WCP algorithm is shown to have linear (in $N$) time complexity.
Computation of \RO\ vector clocks is similar to WCP vector clocks as for both cases we need to track (standard) critical sections
and check if they contain some conflicts.

\section{Experiments}
\label{sec:experiments}

\newcommand{\VOne}{\mbox{SPDOfflineCPP}\xspace}
\newcommand{\VTwo}{\mbox{SPDOfflineCPP$^{\LWLHeldT}$}\xspace}
\newcommand{\VFour}{\mbox{SPDOfflineCPP$^{\ROLHeldT}$}\xspace}

\newcommand{\Benchmark}{\mbox{\bf Benchmark}}
\newcommand{\TT}{\mbox{${\mathcal T}$}}
\newcommand{\EE}{\mbox{${\mathcal E}$}}
\newcommand{\MM}{\mbox{${\mathcal M}$}}
\newcommand{\LL}{\mbox{${\mathcal L}$}}
\newcommand{\Sum}{\mbox{$\sum$}}

\newcommand{\Cycles}{\mbox{Dlk}}      %% Deadlock = Dlk = cycle
\newcommand{\Dependencies}{\mbox{Deps}}
\newcommand{\DepsGuarded}{\mbox{Grds}}
\newcommand{\DepsExtras}{\mbox{E}}
\newcommand{\Time}{\mbox{Time}}
\newcommand{\PhaseOne}{\mbox{P1}}
\newcommand{\PhaseTwo}{\mbox{P2}}
\newcommand{\Races}{\mbox{Races}}
\newcommand{\Guards}{\mbox{Guards}}

%% shorter names for space reasons
\newcommand{\DepExtra}{\mbox{D+}}
\newcommand{\GrdExtra}{\mbox{G+}} %%{\mbox{Guards}}

%%runTableDetails
\begin{table*}[p]
  \caption{
    Columns~2--4 show the number of deadlocks,
    number of lock dependencies and running times
    for \VOne. Times are in seconds, and rounded to the nearest hundreth.
    Column~5--8 show the number of deadlocks,
    number of additionally discovered lock dependencies (\DepExtra),
    number of additional guard locks (\GrdExtra) and running times
    for \VTwo.
    Columns~9--12 show the same for \VFour.
    Running times are factors compared to \VOne.
  }
  \label{tbl:benchmarks}

  \begin{tabular}{|r||r|r|r||r|r|r|r||r|r|r|r|}
 \hline
1 & 2 & 3 & 4 & 5 & 6 & 7 & 8 & 9 & 10 & 11 & 12
 \\
 \hline
\multirow{2}{*}{\Benchmark} &
\multicolumn{3}{c||}{\VOne} &
\multicolumn{4}{c||}{\VTwo} &
\multicolumn{4}{c|}{\VFour} \\ \cline{2-4} \cline{5-8} \cline{9-12}
& \Cycles &  \Dependencies & \Time
& \Cycles & \DepExtra & \GrdExtra & \Time
& \Cycles & \DepExtra & \GrdExtra & \Time
 \\
 \hline
Test-Dimminux & 2 & 7 & 0.00 & 2 & 0 & 0 & 1x & 2 & 0 & 0 & 1x  \\   \hline
StringBuffer & 1 & 3 & 0.00 & 1 & 0 & 0 & 1x & 1 & 0 & 0 & 1x  \\   \hline
Test-Calfuzzer & 1 & 5 & 0.00 & 1 & 0 & 0 & 1x & 1 & 0 & 0 & 1x  \\   \hline
DiningPhil & 1 & 25 & 0.00 & 1 & 0 & 0 & 1x & 1 & 0 & 0 & 1x  \\   \hline
HashTable & 0 & 42 & 0.00 & 0 & 0 & 0 & 1x & 0 & 0 & 0 & 1x  \\   \hline
Account & 0 & 12 & 0.00 & 0 & 0 & 0 & 1x & 0 & 0 & 0 & 1x  \\   \hline
Log4j2 & 0 & 3 & 0.00 & 0 & 0 & 0 & 1x & 0 & 0 & 0 & 1x  \\   \hline
Dbcp1 & 1 & 6 & 0.00 & 1 & 0 & 0 & 1x & 1 & 0 & 0 & 1x  \\   \hline
Dbcp2 & 0 & 18 & 0.00 & 0 & 0 & 0 & 1x & 0 & 0 & 0 & 1x  \\   \hline
Derby2 & 0 & 0 & 0.00 & 0 & 0 & 0 & 1x & 0 & 0 & 0 & 1x  \\   \hline
elevator & 0 & 0 & 0.39 & 0 & 0 & 0 & 1x & 0 & 0 & 0 & 1x  \\   \hline
hedc & 0 & 4 & 0.84 & 0 & 0 & 0 & 1x & 0 & 0 & 0 & 1x  \\   \hline
JDBCMySQL-1 & 0 & 3\mbox{K} & 0.80 & 0 & 0 & 0 & 1x & 0 & 0 & 0 & 1x  \\   \hline
JDBCMySQL-2 & 0 & 3\mbox{K} & 0.81 & 0 & 0 & 0 & 1x & 0 & 0 & 0 & 1x  \\   \hline
JDBCMySQL-3 & 1 & 3\mbox{K} & 0.82 & 1 & 0 & 0 & 1x & 1 & 0 & 0 & 1x  \\   \hline
JDBCMySQL-4 & 1 & 3\mbox{K} & 0.80 & 1 & 0 & 0 & 1x & 1 & 0 & 0 & 1x  \\   \hline
cache4j & 0 & 10\mbox{K} & 1.35 & 0 & 0 & 0 & 1x & 0 & 0 & 0 & 1x  \\   \hline
ArrayList & 4 & 8\mbox{K} & 4.56 & 4 & 0 & 0 & 1x & 4 & 0 & 0 & 3x  \\   \hline
IdentityHashMap & 1 & 80 & 5.22 & 1 & 0 & 0 & 1x & 1 & 0 & 0 & 3x  \\   \hline
Stack & 3 & 95\mbox{K} & 5.59 & 3 & 0 & 0 & 1x & 3 & 0 & 0 & 5x  \\   \hline
LinkedList & 4 & 7\mbox{K} & 6.53 & 4 & 0 & 0 & 1x & 4 & 0 & 0 & 2x  \\   \hline
HashMap & 1 & 4\mbox{K} & 6.93 & 1 & 0 & 0 & 1x & 1 & 0 & 0 & 2x  \\   \hline
WeakHashMap & 1 & 4\mbox{K} & 7.08 & 1 & 0 & 0 & 1x & 1 & 0 & 0 & 2x  \\   \hline
Vector & 1 & 200\mbox{K} & 6.76 & 1 & 0 & 0 & 1x & 1 & 0 & 0 & 1x  \\   \hline
LinkedHashMap & 1 & 4\mbox{K} & 8.54 & 1 & 0 & 0 & 1x & 1 & 0 & 0 & 2x  \\   \hline
montecarlo & 0 & 0 & 16.59 & 0 & 0 & 0 & 1x & 0 & 0 & 0 & 1x  \\   \hline
TreeMap & 1 & 4\mbox{K} & 17.14 & 1 & 0 & 0 & 1x & 1 & 0 & 0 & 1x  \\   \hline
hsqldb & 0 & 125\mbox{K} & 41.36 & 0 & 0 & 0 & 1x & 0 & 0 & 0 & 1x  \\   \hline
sunflow & 0 & 248 & 49.49 & 0 & 0 & 0 & 1x & 0 & 0 & 0 & 1x  \\   \hline
jspider & 0 & 2\mbox{K} & 52.92 & 0 & 0 & 0 & 1x & 0 & 0 & 0 & 1x  \\   \hline
tradesoap & 0 & 40\mbox{K} & 122.73 & 0 & 51 & 104 & 1x & 0 & 51 & 104 & 1x  \\   \hline
tradebeans & 0 & 40\mbox{K} & 122.98 & 0 & 43 & 88 & 1x & 0 & 43 & 88 & 1x  \\   \hline
TestPerf & 0 & 0 & 141.04 & 0 & 0 & 0 & 1x & 0 & 0 & 0 & 1x  \\   \hline
Groovy2 & 0 & 29\mbox{K} & 277.97 & 0 & 3\mbox{K} & 33\mbox{K} & 1x & 0 & 3\mbox{K} & 33\mbox{K} & 1x  \\   \hline
tsp & 0 & 0 & 739.49 & 0 & 0 & 0 & 1x & 0 & 0 & 0 & 1x  \\   \hline
lusearch & 0 & 41\mbox{K} & 541.82 & 0 & 0 & 0 & 1x & 0 & 0 & 0 & 1x  \\   \hline
biojava & 0 & 545 & 512.02 & 0 & 0 & 0 & 1x & 0 & 0 & 0 & 1x  \\   \hline
graphchi & 0 & 82 & 581.99 & 0 & 2 & 2 & 1x & 0 & 2 & 2 & 1x  \\   \hline
 \hline
\Sum & 25 & 626\mbox{K} & 3274.59 & 25 & 3\mbox{K} & 33\mbox{K} & 1x & 25 & 3\mbox{K} & 34\mbox{K} & 1x
 \\  \hline   \end{tabular}

\end{table*}

We evaluated our approach in an offline setting using a large dataset
of pre-recorded program traces from prior work
plus some traces derived from programs such as
the one in~\cref{fig:real}.

\paragraph{Test candidates.}
For experimentation, we consider the following three test candidates:
\begin{description}
\item[\VOne] is our reimplementation of \SPDOffline\ in C++~\footnote{The original \SPDOffline\ is implemented in Java. As we had to reimplement the first phase of \SPDOffline\ to
  integrate refined deadlock patterns, we decided to go for C++ based on personal preference.}
  \VOne\ uses standard deadlock patterns like \SPDOffline.
\item[\VTwo] makes use of \LW\ deadlock patterns in the first phase.
  The second phase is identical to \VOne.
\item[\VFour] makes use of \RO\ deadlock patterns in the first phase.
  The second phase is again identical to \VOne.
\end{description}

\paragraph{Benchmarks and system setup.}

Our (performance) experiments are based on a large set of benchmark traces from
prior work~\cite{conf/pldi/TuncMPV23,10.1145/3503222.3507734}.
These traces are obtained from standard benchmark suites such as
Java Grande~\cite{smith2001parallel}, Da
Capo~\cite{Blackburn:2006:DBJ:1167473.1167488}, and IBM
Contest~\cite{farchi2003concurrent}.
We conducted our experiments on an Apple M4 max CPU with 48GB of RAM
running macOS Sequoia (Version 15.1).

Compared to~\cite{conf/pldi/TuncMPV23}
we excluded five benchmarks (RayTracer, jigsaw, Sor, Swing, eclipse),
because their traces are not well-formed.
%% For example, locks are acquired by distinct threads with no release in between.
The issue has been confirmed by~\citet{conf/pldi/TuncMPV23}.
For space reasons we also had to exclude some of the smaller benchmarks.

\cref{tbl:benchmarks} shows the results
for the standard benchmark traces
that are also used by earlier works~\cite{conf/pldi/TuncMPV23}.

\paragraph{Performance.}

The overall running times of \VOne, \VTwo and \VFour are about the same.
For \VFour there are seven benchmarks (ArrayList -- WeakHashMap, LinkedHashMap)
where we encounter an increase by a factor of two to five.
These benchmarks have a large number of 801~threads.
Detailed benchmark statistics are given in~\cref{sec:bench-details}.
The computation of \RO\ vector clocks is affected by a large number of threads
because we need to track conflicts among critical sections
in different threads to enforce rule~(2) in~\cref{def:release-order}.

\paragraph{Precision.}

For the standard benchmark traces,
all three candidates yield the exact same number of deadlocks.
For four benchmarks (tradesoap, tradebeans, Groovey2, graphchi)
the lock set constructions based on~\LW\ and~\RO\ detect additional locks held
and also reports additional lock dependencies.
Columns~6 and~10 report the additional number of lock dependencies
for \VTwo and \VFour relative to \VOne.
Columns~7 and~11 report the additional number of (guard) locks
that are in $\LWLHeldT$ and $\ROLHeldT$ compared to $\TOLHeldT$
when computing lock dependencies.
For benchmark Groovey2, \VFour reports a few more guard locks and lock dependencies.
Due to rounding the difference can only be observed
when comparing the overall numbers.

To show improvements in terms of precision, we applied all three test
candidates on traces derived some C programs.
The trace derived from the program in~\cref{fig:real}
effectively corresponds to the trace~\cref{fig:pseudoDPFalseNegative}.
For this trace, candidates \VTwo and \VFour report (correctly) a deadlock
whereas \VOne fails to detect the deadlock.
It is easy to derive variations of~\cref{fig:pseudoDPFalseNegative}
including the corresponding traces to show that \VFour detects more deadlocks than
\VTwo. Such traces are similar to trace~\traceRef{TraceCSRO}
in~\cref{fig:ex-cs-ro}. See \cref{sec:extra-traces}.

\section{Related Work}
\label{sec:rw}

% Havelund~
\citet{10.5555/645880.672085} and
% Harrow~
\citet{DBLP:conf/spin/Harrow00}
were the first to identify deadlocking situations via a circularity in the lock order dependency among threads.
% Bensalem and Havelund~
\citet{conf/hvc/BensalemH05}
make explicit use of lock sets to formalize lock-order graphs that capture the lock order dependency for the threads.
\citet{conf/pldi/JoshiPSN09}
introduce lock dependencies on a per-thread basis.
The advantage compared to lock-order graphs is that locks
resulting from the same thread can be more easily detected
compared to the lock-order graph representation.
A deadlock warning is issued if we find
a cyclic chain of lock dependencies
satisfying conditions \textbf{DP-Cycle} and \textbf{DP-Guard}.
There are numerous works~\cite{10.1145/3377811.3380367,Samak:2014:TDD:2692916.2555262,conf/ase/ZhouSLCL17,6718069,conf/oopsla/KalhaugeP18,conf/fse/CaiYWQP21,conf/pldi/TuncMPV23} that are built upon this idea.
All of the above works rely on the standard lock lock set construction.

The technical report~\cite{report/MathurPTV23}
of the conference paper~\cite{conf/pldi/TuncMPV23}
gives a Java program that corresponds
to~\cref{fig:pseudoDPFalsePositive} to show that
the DIRK deadlock predictor~\cite{conf/oopsla/KalhaugeP18} yields false positives
but does not define multi-thread critical sections,
does not analyze false negatives, and
does not propose corrected definitions of critical sections and lock sets.

\section{Conclusion}
\label{sec:concl}

Pretty much all prior dynamic deadlock detection methods assume per-thread lock sets
and critical sections restricted to a single thread. This assumption is simply false.
Critical sections that cover multiple threads exist,
can be observed in simple programs and invalidate the standard foundation.
We give corrected definitions based on a trace-based characterization of critical sections and
lock sets to deal with programs where locks protect events across thread boundaries.
%% This is not a new language feature we are proposing.
%% Such cases can arise in actual programs.
%% See the examples from the introduction
%% and our microbenchmarks in C.

We introduce some underapproximations of the general construction
that can be implemented efficiently via the use of partial order methods.
We study their theoretical properties such as the impact
on the precision of lock set-based deadlock precision.
We can guarantee that predictable deadlock patterns
employing the partial order-based constructions are still true positives (\cref{le:stuckness}).

Our contribution is orthogonal to any particular deadlock predictor.
Computation of more general deadlock patterns as described
by~\cref{alg:lock-deps} can be integrated into
iGoodLock~\cite{conf/pldi/JoshiPSN09},
UNDEAD~\cite{conf/ase/ZhouSLCL17},
SeqCheck~\cite{conf/fse/CaiYWQP21} and others.
We choose the \SPDOffline\ deadlock predictor
because~\cref{alg:dp-syncp} removes false positives,
and \SPDOffline\ is the standard sound predictor operating under this model.
Hence, we integrate the partial order-based constructions into
the \SPDOffline\ deadlock predictor.
The resulting extensions of \SPDOffline\ remain sound (no false positives)
and we are able to cover a larger class of deadlocks (fewer false negatives).
Our  extensive set of performance benchmarks highlight two key findings:
Generalized lock sets are practically feasible with no performance penalty, and
additional lock dependencies do arise, even though they do not yield additional
deadlocks on this benchmark set.
For traces derived from C programs we can show that the precision (fewer false negatives) is improved.

There are several avenues for future work.
We follow the standard fixed-reads-from assumption
under the sequential consistency model used in all recent
predictive analyses; e.g.~consider~\cite{Smaragdakis:2012:SPR:2103621.2103702,conf/pldi/KiniMV17,conf/fse/CaiYWQP21,conf/pldi/TuncMPV23}.
Tracking (written/read) values and covering more relaxed memory models
is an interesting topic for future work
but requires substantial changes to the theory and implementation.
Other directions include the development of
further partial orders that are suitable for our purposes
and to investigate the limitations
of per-thread lock sets in the context of data race prediction methods.
%% MS:
%% check out static deadlock prediction.

\begin{acks}

We thank some ESOP’26 and OOPSLA'26 referees for their comments on a previous version of this paper.

\end{acks}

%% There are a number of avenues for future work.
%% There are numerous works~\cite{Savage:1997:EDD:265924.265927,10.1145/781498.781528,serebryany2009threadsanitizer,Roemer:2018:HUS:3296979.3192385,Andreas:Pavlogiannis:popl20,10.1145/3360605,conf/mplr/SulzmannS20} that
%% employ lock sets for dynamic data race prediction.
%%
%% \ms{mention partial orders relying on critical sections}
%% WCP, SDP, WDP, PWR, ...
%% They all assume standard critical sections ...
%%
%%
%%
%%
%% \begin{verbatim}
%% Impact of unbounded on SPD.
%% Remains sound but more complete.
%%
%% Impact on DIRK.
%% More sound, still not strong enough, is unbounded-CMHB strong enough?
%%
%% Unbounded-LW still too weak to fix DIRK.
%% Recall LW <= CMHB.
%% => future work
%%
%%  CMHB <= TWR   too strong, false positives
%%
%%  what about
%%
%%  CMHB <= MHB_from_oopsla23   also too strong
%%
%%
%% \end{verbatim}

% \bibliographystyle{plainnat}
%% \bibliographystyle{splncs04nat} % https://github.com/tpavlic/splncs04nat

\clearpage
\appendix

%--------------------------------------------------
%--------------------------------------------------
\section{Proofs}
\label{sec:Proofs}

%--------------------------------------------------
\subsection{Proof of \cref{le:stuckness}}

\stuckness*

%% MS:
%% We only require (1b) and (2b)!
%% Given there's a witness, we don't need to exploit DP-Cycle!
\begin{proof}
  Condition \textbf{WF-Req} in \cref{sec:prelim:wf}
  guarantees that
  (a) each request~$q_i$ is immediately succeeded
  by some acquire~$a_i$ or
  (b) $q_i$ is the final event in thread~$\thd(q_i)$ in trace~$T$.
  In case of (b), we simply assume some acquire~$a_i$ that fulfills the request $q_i$.

  For brevity we only consider the case $\SET{q_1,q_2}$.
  W.l.o.g., $q_i = (\_, \reqE{l_i})$.
  Then, the matching acquire is of the following form $a_i = (\_,\lockE{l_i})$.
  The thread ids do not matter here. So, we ignore them by writing $\_$.

  By assumption (1) $l_1 \in \GeneralLHeld{T}{q_2}$ and (2) $l_2 \in \GeneralLHeld{T}{q_1}$.

  From (1) and (2) we derive that
  (1a) $a \crpLt[T] q_2$, (1b) $q_2 \crpLt[T] r$,
  (2a) $a' \crpLt[T] q_1$ and (2b) $q_1 \crpLt[T] r'$
  where $a=(\_,\lockE{l_1})$, $r=(\_,\unlockE{l_1})$ are the acquire/release events
  that hold lock~$l_1$ and
  $a'=(\_,\lockE{l_2})$, $r'=(\_,\unlockE{l_2})$ are the acquire/releave events that
  hold lock~$l_2$.

  W.l.o.g., we can assume that (3) $T'=T_1 \conc [q_1,q_2]$ for some $T_1$.
  Requests $q_i$ are final. So, we can swap them with a later event in the trace
  and reach the form~(3).

  Suppose there exists $T''$ such that
  $T' \conc T'' \conc [a_1]$ is well-formed.
  From (1a) and (3) we derive that $a \in T'$ and $a \TrLt[T'] q_2$.
  But this means that the release event $r$ that corresponds to~$a$
  must occur before~$a_1$ (either in $T'$ or $T''$).
  Otherwise, lock semantics is violated.

  Suppose $r \in T'$. Via (1b) we find that $q_2 \in T'$ and $q_2 \TrLt[T'] r$.
  This is a contradiction to (3).
  Hence, $r \in T''$ where we find (4) $T'' = T_2 \conc [r] \conc T_3$ for some $T_2$ and $T_3$.

  From (3) and (4) we arrive at the well-formed trace
  $$
 (5) ~~    T' \conc T'' \conc [a_1] = T_1 \conc [q_1,q_2] \conc T_2 \conc [r] \conc T_3 \conc [a_1]
  $$
  We continue by distinguishing among the following cases.
  Each case exploits some assumption to move a request towards the end of the trace.
  Each case leads to a contradiction.

  \textbf{Case $a_2 \not \in T_2 \conc [r] \conc T_3$:}

  This means that the acquire~$a_2$ that fulfills~$q_2$ is not present.
  Under this assumption, request $q_2$ can be swapped with later events in the trace
  while maintaining well-formedness.
  Hence, from (5) we obtain the well-formed trace
  $T''' = T_1 \conc [q_1] \conc T_2 \conc [r] \conc T_3 \conc [a_1,q_2]$.
  This shows that $r \TrLt[T'''] q_2$ and contradicts (1b).

  \textbf{Case $a_2 \in T_2 \conc [r] \conc T_3$:}

  Under this assumption we find that in the well-formed trace
  $T_1 \conc [q_1,q_2] \conc T_2 \conc [r] \conc T_3 \conc [a_1]$
  either $a_2 \in T_2$ or $a_2 \in T_3$.

  From (2a) and (3) we find that $a' \in T_1$ and $a' \TrLt[T_1] q_2$.
  But this means that the release event $r'$ that corresponds to $a'$
  must occur before $a_2$. Otherwise, lock semantics is violated.

  Suppose $a_2 \in T_2$. Then, $r'$ before $a_2$ in $T_2$.
  W.l.o.g., $T_2 = [r',a_2'] \conc T_2'$ for some $T_2'$.
  The (well-formed) trace is as follows.
  $$
  T_1 \conc [q_1,q_2] \conc [r',a_2'] \conc T_2' \conc [r] \conc T_3 \conc [a_1]
  $$
  Request~$q_1$ can be moved closer to its fulfilling acquire~$a_1$ while maintaining
  well-formedness. In particular, request~$q_1$ can be moved past~$r'$.
  This contradicts (2b).

  The other case where $a_2 \in T_3$ applies similar arguments
  and also leads to a contradiction.

  We summarize.
  All cases lead to a contradiction
  and therefore there is no  $T''$ such that
  $T' \conc T'' \conc [a_1]$ is well-formed.
  Similar argument applies to $a_2$.
\end{proof}

%--------------------------------------------------
\subsection{Proof of \cref{le:mhb-criteria-inclusion}}

\mhbcriteriainclusion*

\begin{proof}
  Suppose $l \in \LH{\pLt[T]}{e}$.
  Then, $e \in \CS{\pLt[T]}{a,r}$.
  By condition \textbf{\PO-CS-Enclosed}
  we have that $a \pLt[T] e$ and $e \pLt[T] r$.
  By \textbf{MHB-Criteria} all $T' \in \crp(T)$ we find
  that $a \TrLt[T'] e$ and $e \TrLt[T'] r$.
  This shows that condition \textbf{CS-Enclosed} holds.
  Condition \textbf{\PO-CS-Match} and \textbf{CS-Match} are the same.
  Hence, $e \in \GeneralCSect{T}{a,r}$
  and we can conclude that $l \in \GeneralLHeld{T}{e}$.
\end{proof}

%--------------------------------------------------
\subsection{Proof of \cref{le:pd-mhb}}

\lepdmhb*

\begin{proof}
We prove the result by contradiction.
For brevity, we only consider the case
for deadlock patterns are of size two.

We assume that
(1) $\SET{q_1,q_2} \in \AllPD{\LHH{\pLt[T]}}$ and
(2) $\SET{q_1,q_2} \not\in \AllPD{\GeneralLHeldd{T}}$.

The `witness' cannot be the issue.
Hence, we have that
(3) $\SET{q_1,q_2} \in \AllDP{\LHH{\pLt[T]}}$ and
(4) $\SET{q_1,q_2} \not\in \AllDP{\GeneralLHeldd{T}}$.
where (4) does not hold because \textbf{DP-Cycle} or \textbf{DP-Guard} are not satisfied.
Condition \textbf{DP-Cycle} can not be responsible because of \cref{le:mhb-criteria-inclusion}.
Hence, the issue is \textbf{DP-Guard}.

There must exist (3) $l \in \LH{\pLt[T]}{q_1} \indexedcap \LH{\pLt[T]}{q_2}$.
This means that there exist two distinct acquires~$a$ and~$a'$
on lock~$l$ such that
(4) $a \pLt[T] q_1$, $a' \pLt[T] q_2$, $q_1 \pLt[T] r$ and $q_2 \pLt[T] r'$
where $r,r'$ are the matching release partners of $a,a'$.

From (1) we conclude that there exists a witness $T'\in \crp(T)$
where $q_1$ and $q_2$ are final in $T'$.
From (4) and the \textbf{MHB-Criteria} assumption
we conclude that $a$ and $a'$ appear before $q_1$ and $q_2$ in $T'$
and therefore we find that (5) $a \TrLt[T'] q$ and $a' \TrLt[T'] q$ for $q \in \{q_1,q_2\}$.
To ensure that lock semantics is not violated
at one of the release events $r,r'$ must appear in between $a$ and $a'$.

W.l.o.g., we assume $a$ appears before $a'$
and then we find that $a \TrLt[T'] r \TrLt[T'] a'$.
This leads to a contraction via (4)
and the assumption that $\pLt[T]$ satisfies the \textbf{MHB-Criteria}.
\end{proof}

%--------------------------------------------------
\subsection{Proof of \cref{le:ro-mhb}}

\leromhb*

\begin{proof}
  (0) $\toLTSym \subseteq \lwLTSym \subseteq \roLTSym$ holds by definition.

  Below, we show that $\roLTSym$ satisies the \textbf{MHB-Criteria}.
  From (0) we can then immediately conclude that
  $\toLTSym$ and $\lwLTSym$ satisfy the \textbf{MHB-Criteria}
  as well.

We show that $\roLTSym$ satisies the \textbf{MHB-Criteria}
by induction over the derivation~$\roLTSym$.

For subcases thread-order and last-write order the \textbf{MHB-Criteria} is immediately met.

Consider $r \roLt[T] f$ if
(1) $e \in \TOCSect{T}{\AcqRelPair{a}{r}}{l}$ and
  (2) $f\in \TOCSect{T}{\AcqRelPair{a'}{r'}}{l}$
and (3) $e \lwLt[T] f$.

Suppose $T' \in \crp(T)$ where $f \in T'$.
From (1) and by induction we derive that (4) $a \TrLt[T'] e \TrLt[T'] r$.
Similarly, we find that (5) $a' \TrLt[T'] f \TrLt[T'] r'$
and (6) $e \TrLt[T'] f$.
From (4) and (6) we derive that acquire~$a$ must appear before~$f$ in $T'$.
Trace $T'$ is well-formed.

Acquires~$a$, $a'$ and releases~$r$, $r'$ operate on the same lock~$l$.
Hence, release~$r$ must appear before~$a'$ as otherwise lock semantics is violated.

This shows that $r\in T'$ and $r \TrLt[T'] a' \TrLt[T'] f$
and concludes the proof.
\end{proof}

%--------------------------------------------------
%--------------------------------------------------
\section{Sync-Preserving Deadlocks}

We repeat some definitions that are introduced by~\citep{conf/pldi/TuncMPV23}.

\begin{definition}[Sync Preservation~{\citep[Definition~1]{conf/pldi/TuncMPV23}}]
  We say $T' \in \crp(T)$ is \emph{sync preserving} if,
  for every distinct acquires~$a,a' \in T'$ such that $\lock(a) = \lock(a')$,
  $a \TrLt[T'] a'$ if and only if~$a \TrLt[T] a'$.
  We define $\spCrp(T) = \{ T' \in \crp(T) \mid \text{$T'$ is sync preserving} \}$.
\end{definition}

\noindent
In words, a correctly reordered prefix is sync preserving if critical sections protected by the same lock appear in source-trace order.
Based on sync-preserving reorderings we can refine
the notion of a predictable deadlock.

We consider a specific class of predictable deadlocks
where the witness is sync-preserving.
As before in \cref{def:pred-deadlock},
the following definition is parametric in terms of lock set function~$L$.

\begin{definition}[Sync-Preserving Predictable Deadlock]
  \label{def:syncp-pred-deadlock}
  Given $\SET{q_1,\dots ,q_n} \in \AllDP{T}{L}$,
  we say $\SET{q_1,\dots ,q_n}$ is a \emph{sync-preserving predictable deadlock} in~$T$ if
  there exists \emph{witness}~$T' \in \spCrp(T)$
  such that each of $q_1,...,q_n$ are final in~$T'$.

  We use notation $\SPD{T}{L}{\SET{q_1,\dots ,q_n}}$ to indicate
  that $\SET{q_1,\dots ,q_n}$ satisfies the above conditions.
\end{definition}
We define the set of sync-preserving predictable deadlocks for trace $T$ induced by $L$ as follows:
$\AllSPD{T}{L} = \{ \SET{q_1,\dots ,q_n} \mid \SPD{T}{L}{\SET{q_1,\dots ,q_n}} \}$.

\begin{figure}[t]

  \begin{subfigure}[t]{0.44\textwidth}
 \bda{l}
%%latexTrace $ addLoc ex_mt3c
\ba{|l|l|l|l|l|}
\hline
\traceNew{SyncDLExample} & \thread{1} & \thread{2} & \thread{3} & \thread{4}\\ \hline
\eventE{1}  & \lockE{\LKA}&&&\\
\eventE{2}  & \writeE{\VA}&&&\\
\eventE{3}  & &\readE{\VA}&&\\
\eventE{4}  & &\reqLockE{\LKB}&&\\
\eventE{5}  & &\lockE{\LKB}&&\\
\eventE{6}  & &\unlockE{\LKB}&&\\
\eventE{7}  & &\writeE{\VB}&&\\
\eventE{8}  & \readE{\VB}&&&\\
\eventE{9}  & \unlockE{\LKA}&&&\\
\eventE{10}  & &&\lockE{\LKA}&\\
\eventE{11}  & &&\unlockE{\LKA}&\\
\eventE{12}  & &&&\lockE{\LKB}\\
\eventE{13}  & &&&\reqLockE{\LKA}\\
\eventE{14}  & &&&\lockE{\LKA}\\
\eventE{15}  & &&&\unlockE{\LKA}\\
\eventE{16}  & &&&\unlockE{\LKB}\\

\hline \ea{}
\\ \mbox{}
\\ \mbox{Sync-preserving witness}
\\
\ba{|l|l|l|l|l|}
\hline
\traceNew{SyncDLExampleWitness} & \thread{1} & \thread{2} & \thread{3} & \thread{4}\\ \hline
\eventE{1}  & \lockE{\LKA}&&&\\
\eventE{2}  & \writeE{\VA}&&&\\
\eventE{3}  & &\readE{\VA}&&\\
\eventE{12}  & && \hphantom{\lockE{\LKB}} &\lockE{\LKB}\\
\eventE{4}  & &\reqLockE{\LKB}&&\\
\eventE{13}  & &&&\reqLockE{\LKA}\\

 \hline \ea{}

 \eda
  \caption{Sync-preserving deadlock.}
 \label{fig:syncp-dl}
  \end{subfigure}
  \qquad
  \begin{subfigure}[t]{0.44\textwidth}
    \bda{l}
%%latexTrace $ addLoc ex_mt3d
\ba{|l|l|l|l|}
\hline
\traceNew{NotSyncDLExample} & \thread{1} & \thread{2} & \thread{3}\\ \hline
\eventE{1}  & \lockE{\LKA}&&\\
\eventE{2}  & \writeE{\VA}&&\\
\eventE{3}  & &\readE{\VA}&\\
\eventE{4}  & &\reqLockE{\LKB}&\\
\eventE{5}  & &\lockE{\LKB}&\\
\eventE{6}  & &\unlockE{\LKB}&\\
\eventE{7}  & &\writeE{\VB}&\\
\eventE{8}  & \readE{\VB}&&\\
\eventE{9}  & \unlockE{\LKA}&&\\
\eventE{10}  & &&\lockE{\LKA}\\
\eventE{11}  & &&\unlockE{\LKA}\\
\eventE{12}  & &&\lockE{\LKB}\\
\eventE{13}  & &&\reqLockE{\LKA}\\
\eventE{14}  & &&\lockE{\LKA}\\
\eventE{15}  & &&\unlockE{\LKA}\\
\eventE{16}  & &&\unlockE{\LKB}\\

\hline \ea{}
\\ \mbox{}
\\ \mbox{Witness but not sync-preserving}
\\
\ba{|l|l|l|l|}
\hline
\traceNew{NotSyncDLExampleWitness} & \thread{1} & \thread{2} & \thread{3}\\ \hline
\eventE{10}  & &&\lockE{\LKA}\\
\eventE{11}  & &&\unlockE{\LKA}\\
\eventE{1}  & \lockE{\LKA}&&\\
\eventE{2}  & \writeE{\VA}&&\\
\eventE{3}  & &\readE{\VA}&\\
\eventE{12}  & &&\lockE{\LKB}\\
\eventE{4}  & &\reqLockE{\LKB}&\\
\eventE{13}  & &&\reqLockE{\LKA}\\

\hline \ea{}
 \eda
  \caption{Predictable but not sync-preserving deadlock.}
 \label{fig:non-syncp-dl}
  \end{subfigure}
    \caption{Not all predictable deadlocks are sync-preserving.}
\label{fig:syncp-and-non-syncp-dls}
\end{figure}

\traceRef{SyncDLExample} in \cref{fig:syncp-dl}
contains a sync-preserving (predictable) deadlock.
We find that $\SET{e_4, e_{13}} \in \AllSPD{\traceRef{SyncDLExample}}{\MultiLHeldIdxT}$.
However, not all deadlocks are sync-preserving
as shown in~\cref{fig:non-syncp-dl}.

\section{Sync-Preserving Deadlock Prediction}

\begin{definition}[Sync-preserving Closure~{\citep[Definition~3]{conf/pldi/TuncMPV23}}]
  \label{def:spClosure}
  Let $T$ be a trace and $S \subseteq \evts(T)$.
  The \emph{sync-preserving closure} of~$S$ in~$T$, denoted $\spClosure[T](S)$, is the smallest set $S'$ such that
 \textbf{(SPC-Closure)} $S \subseteq S''$,
 \textbf{(SPC-LW)} for every~$e, f \in T$ such that $e \lwLt[T] \vEvtB$, $f \in S'$ implies~$e \in S'$,
      and
  \textbf{(SPC-SP)} for every distinct acquires~$a, a' \in T$ with~$\lock(a) = \lock(a')$, $a \TrLeq[T] a'$ implies~$\rel[T](a) \in S'$.
\end{definition}

\noindent
We generally assume that $S$ equals the set $\SET{q_1,\dots ,q_n}$ of requests that characterize
a deadlock pattern.
\textbf{SPC-Closure} makes sure that at the least the requests are included.
\textbf{SPC-LW} ensures that thread order and last-write dependencies are satisfied.
\textbf{SPC-SP} makes sure that sync-preservation holds
by ensuring that earlier critical sections are closed before later ones are opened.

The following result shows that the sync-preserving closure determines exactly the events necessary to form a sync-preserving reordering.

\begin{lemma}[{\citep[Lemma~4.1]{conf/pldi/TuncMPV23}}]
  \label{lem:spClosure}
  Let $T$ be a trace and let $S \subseteq \evts(T)$.
  There exists~$T' \in \spCrp(T)$ such that $\evts(T') = \spClosure[T](S)$.
\end{lemma}

This then leads to an elegant method to verify if a deadlock pattern
implies a sync-preserving deadlock.
Build the sync-preserving closure of requests
and check that none of the matching acquires are included.

\begin{lemma}
  \label{le:multi-sync-sound}
  Let $\DP{T}{L}{\SET{q_1,\dots ,q_n}}$ be a deadlock pattern.
  $\SET{q_1,\dots ,q_n}$ is a sync-preserving predictable deadlock iff
  $\spClosure[\vTr](\SET{q_1,\dots ,q_n}) \cap \{ \acq[T](q) \mid q \in \SET{q_1,\dots ,q_n} \} = \emptyset$.
\end{lemma}

The above result generalizes the earlier~\citep[Lemma~4.2]{conf/pldi/TuncMPV23}.
Lemma~4.2 only covers the case where $L$ equals the standard lock set function~$\TOLHeldT$.
The proof of Lemma~4.2.~can be naturally extended to
deadlock patterns and sync-preserving predictable deadlock
induced by~$\GeneralLHeldd{T}$ and its approximations.

In combination with results~\cref{le:stuckness} and~\cref{le:pd-to-lw-ro}
this shows that we can cover a larger class of sync-preserving deadlocks
and no false positives are reported.

\section{Benchmark Details}
\label{sec:bench-details}

See \cref{tbl:stats}.

%%runStats
\begin{table*}[t]
    \caption{
    \textbf{Trace statistics}.
    Columns~2--5 contain the number of events, total number of threads.
  }
  \label{tbl:stats}
\begin{tabular}{|r|r|r|r|r|}
 \hline
1 & 2 & 3 & 4 & 5
 \\
 \hline
\Benchmark &
\EE &
\TT &
\MM &
\LL
 \\
 \hline
Test-Dimminux & 50 & 3 & 8 & 6  \\   \hline
StringBuffer & 57 & 3 & 13 & 3  \\   \hline
Test-Calfuzzer & 126 & 5 & 15 & 5  \\   \hline
DiningPhil & 210 & 6 & 20 & 5  \\   \hline
HashTable & 222 & 3 & 4 & 2  \\   \hline
Account & 617 & 6 & 46 & 6  \\   \hline
Log4j2 & 1\mbox{K} & 4 & 333 & 10  \\   \hline
Dbcp1 & 2\mbox{K} & 3 & 767 & 4  \\   \hline
Dbcp2 & 2\mbox{K} & 3 & 591 & 9  \\   \hline
Derby2 & 3\mbox{K} & 3 & 1\mbox{K} & 3  \\   \hline
elevator & 222\mbox{K} & 5 & 726 & 51  \\   \hline
hedc & 410\mbox{K} & 7 & 109\mbox{K} & 7  \\   \hline
JDBCMySQL-1 & 436\mbox{K} & 3 & 73\mbox{K} & 10  \\   \hline
JDBCMySQL-2 & 436\mbox{K} & 3 & 73\mbox{K} & 10  \\   \hline
JDBCMySQL-3 & 436\mbox{K} & 3 & 73\mbox{K} & 12  \\   \hline
JDBCMySQL-4 & 437\mbox{K} & 3 & 73\mbox{K} & 13  \\   \hline
cache4j & 758\mbox{K} & 2 & 46\mbox{K} & 19  \\   \hline
ArrayList & 3\mbox{M} & 801 & 121\mbox{K} & 801  \\   \hline
IdentityHashMap & 3\mbox{M} & 801 & 496\mbox{K} & 801  \\   \hline
Stack & 3\mbox{M} & 801 & 118\mbox{K} & 2\mbox{K}  \\   \hline
LinkedList & 3\mbox{M} & 801 & 290\mbox{K} & 801  \\   \hline
HashMap & 3\mbox{M} & 801 & 555\mbox{K} & 801  \\   \hline
WeakHashMap & 3\mbox{M} & 801 & 540\mbox{K} & 801  \\   \hline
Vector & 3\mbox{M} & 3 & 14 & 3  \\   \hline
LinkedHashMap & 4\mbox{M} & 801 & 617\mbox{K} & 801  \\   \hline
montecarlo & 8\mbox{M} & 3 & 850\mbox{K} & 2  \\   \hline
TreeMap & 9\mbox{M} & 801 & 493\mbox{K} & 801  \\   \hline
hsqldb & 20\mbox{M} & 46 & 945\mbox{K} & 402  \\   \hline
sunflow & 21\mbox{M} & 15 & 2\mbox{M} & 11  \\   \hline
jspider & 22\mbox{M} & 11 & 5\mbox{M} & 14  \\   \hline
tradesoap & 42\mbox{M} & 236 & 3\mbox{M} & 6\mbox{K}  \\   \hline
tradebeans & 42\mbox{M} & 236 & 3\mbox{M} & 6\mbox{K}  \\   \hline
TestPerf & 80\mbox{M} & 50 & 598 & 8  \\   \hline
Groovy2 & 120\mbox{M} & 13 & 13\mbox{M} & 10\mbox{K}  \\   \hline
tsp & 307\mbox{M} & 10 & 181\mbox{K} & 2  \\   \hline
lusearch & 217\mbox{M} & 10 & 5\mbox{M} & 118  \\   \hline
biojava & 221\mbox{M} & 6 & 121\mbox{K} & 78  \\   \hline
graphchi & 216\mbox{M} & 20 & 25\mbox{M} & 60  \\   \hline
 \hline
\Sum & 1354\mbox{M} & 7\mbox{K} & 61\mbox{M} & 30\mbox{K}
 \\  \hline   \end{tabular}
\end{table*}

\section{Traces Derived from C Programs}
\label{sec:extra-traces}

\begin{figure}

  \begin{subfigure}[t]{0.39\textwidth}
    \bda{l}
%%latexTrace $ addLoc mb_fn_2
\ba{|l|l|l|l|}
\hline
\traceNew{TraceCSROEXT} & \thread{1} & \thread{2} & \thread{3}\\ \hline
\eventE{1}  & \forkE{\thread{2}}&&\\
\eventE{2}  & \forkE{\thread{3}}&&\\
\eventE{3}  & \lockE{\LKB}&&\\
\eventE{4}  & \writeE{\VA}&&\\
\eventE{5}  & \lockE{\LKA}&&\\
\eventE{6}  & \unlockE{\LKB}&&\\
\eventE{7}  & &\lockE{\LKB}&\\
\eventE{8}  & &\readE{\VA}&\\
\eventE{9}  & &\unlockE{\LKB}&\\
\eventE{10}  & &\lockE{\LKC}&\\
\eventE{11}  & &\unlockE{\LKC}&\\
\eventE{12}  & &\writeE{\VB}&\\
\eventE{13}  & \readE{\VB}&&\\
\eventE{14}  & \unlockE{\LKA}&&\\
\eventE{15}  & &&\lockE{\LKC}\\
\eventE{16}  & &&\lockE{\LKA}\\
\eventE{17}  & &&\unlockE{\LKA}\\
\eventE{18}  & &&\unlockE{\LKC}\\

 \hline \ea{}
    \eda
    \caption{Extennsion of~\cref{fig:ex-cs-ro}.}
 \label{fig:ext-ex-cs-ro}
  \end{subfigure}
  \qquad
  \begin{subfigure}[t]{0.39\textwidth}
\bda{l}
%%latexTrace $ addLoc mb_fn_3
\ba{|l|l|l|l|l|}
\hline
\traceNew{TraceCSROBYNDEXT}& \thread{1} & \thread{2} & \thread{3} & \thread{4}\\ \hline
\eventE{1}  & \forkE{\thread{2}}&&&\\
\eventE{2}  & \forkE{\thread{3}}&&&\\
\eventE{3}  & \forkE{\thread{4}}&&&\\
\eventE{4}  & \lockE{\LKB}&&&\\
\eventE{5}  & \writeE{\VC}&&&\\
\eventE{6}  & \lockE{\LKA}&&&\\
\eventE{7}  & \unlockE{\LKB}&&&\\
\eventE{8}  & &\lockE{\LKB}&&\\
\eventE{9}  & &\writeE{\VB}&&\\
\eventE{10}  & &&\readE{\VB}&\\
\eventE{11}  & &&\readE{\VC}&\\
\eventE{12}  & &&\lockE{\LKC}&\\
\eventE{13}  & &&\unlockE{\LKC}&\\
\eventE{14}  & &&\writeE{\VA}&\\
\eventE{15}  & &\readE{\VA}&&\\
\eventE{16}  & &\lockE{\LKB}&&\\
\eventE{17}  & &\unlockE{\LKC}&&\\
\eventE{18}  & &\writeE{\VA}&&\\
\eventE{19}  & \readE{\VA}&&&\\
\eventE{20}  & \unlockE{\LKA}&&&\\
\eventE{21}  & &&&\lockE{\LKC}\\
\eventE{22}  & &&&\lockE{\LKA}\\
\eventE{23}  & &&&\unlockE{\LKA}\\
\eventE{24}  & &&&\unlockE{\LKC}\\
 \hline \ea{}

     \eda
    \caption{Extension of~\cref{fig:ex-cs-not}.}
 \label{fig:ext-ex-cs-not}
  \end{subfigure}

  \caption{Traces derived from some C programs.}
\label{fig:ext-ex-cs}
\end{figure}

Consider \cref{fig:ext-ex-cs}.
For trace~\traceRef{TraceCSROEXT} \VTwo\ fails to report the deadlock.
The deadlock is reported by~\VFour.
For trace~\traceRef{TraceCSROBYNDEXT} \VFour\ fails to report the deadlock.
The examples show that the inclusion (b) in \cref{le:pd-to-lw-ro}
is strict.
The corresponding C programs are part of the artifact.

%% \section{TRASH}
%%
%%
%% \begin{verbatim}
%%
%% The story as follows.
%%
%% Stanard lock set construction is restricted.
%% Turns out it's the most restricted variant
%% cause simple to implement.
%%
%% Lock sets and critical sections.
%% Give (finally) an unrestricted definition
%% that captures their meaning.
%%
%% General lock set case, costly to compute.
%% We consider several restrictions.
%% Trade-off between being efficient and complete.
%%
%%
%% ASPLOS 26  (August)
%% OOPSPLA 26 (October)
%%
%% Go straight to TOPLAS?
%%
%%
%% Theory part.
%%
%%   Add complexity.
%%
%%   Cover more
%%
%% Implementation
%%
%%     Overhead of standard vs unrestricted lock set.
%%     Include measurements.
%%    "plug and play", we use SeqP just as an example.
%%
%%
%% Experiments.
%%
%%   Have an artifact.
%%
%%
%% Theory part.
%%   Cover partial orders and how they are affected by
%%   restricted lockset + critical sections.
%%
%%
%% Impact on data race prediction.
%% That's a separate paper but we do some groundwork here.
%%
%%
%% example
%%
%%  monitors with nested locks
%%    where minotor microservices carry out calcs
%%
%%    acq(m)
%%
%%       async(m1)
%%
%%       async(m2)
%%
%%       wait for m1 and m2 to be done
%%
%%    rel(m)
%%
%%
%% SyncP and Go
%%
%%    we don't care about  non-atomics,
%%    they don't imply any HB relations
%%
%%     we enforce "reads-from" dependencies
%%
%%        atomics
%%        channels
%%
%%
%%
%% \end{verbatim}

\end{document}